\documentclass[a4paper,11pt]{article}

\usepackage{authblk}
\usepackage[utf8]{inputenc}

\usepackage{amssymb}
\usepackage{graphicx}
\usepackage{amsthm}         

\usepackage{mathtools}      
\usepackage{upgreek}        
\usepackage{psfrag}         
\usepackage{multirow}       
\usepackage{xargs}          
\usepackage{color}          
\usepackage{lineno}

\usepackage[labelformat=simple]{subcaption} 
\usepackage[pdftex]{hyperref}
\usepackage{kvoptions-patch}

\graphicspath{{fig/}{./fig/}}

\newcommand{\TheTitle}{Capturing points with a rotating polygon\\
  (and a 3D extension)}

\AtBeginDocument{%
  \hypersetup{%
    pdftitle = {\TheTitle}
  }
}

\def\q5uad{\quad\quad\quad\quad\quad}

\usepackage{mathabx}
\def\Hur{H_{\reflectbox{\rotatebox[origin=c]{270}{$\Lsh$}}}}
\def\Hdl{H_{\reflectbox{\rotatebox[origin=c]{90}{$\Lsh$}}}}

\newtheorem{theorem}{Theorem}
\newtheorem{problem}[theorem]{Problem}
\newtheorem{lemma}[theorem]{Lemma}
\newtheorem{corollary}[theorem]{Corollary}

\newcommandx*{\ch}[1][usedefault, 1=P]{\mathcal{CH}({#1})}
\newcommandx*{\cp}[1][usedefault, 1=p]{c_{#1}}
\newcommand{\rp}{r}

\begin{document}


  \title{\TheTitle}

  %
\author[1]{
Carlos Alegr\'{i}a-Galicia
}

\author[2]{
David Orden
}

\author[3]{
Leonidas Palios
}

\author[4]{
Carlos Seara
}

\author[5]{
Jorge Urrutia
}

\affil[1]{
 Posgrado en Ciencia e Ingenier\'{i}a de la    Computaci\'{o}n, Universidad Nacional Aut\'{o}noma de M\'{e}xico.
 \texttt{calegria@uxmcc2.iimas.unam.mx}
 }

\affil[2]{Departamento de F\'{i}sica y Matem\'{a}ticas,
    Universidad de Alcal\'{a}, Spain.
    \texttt{david.orden@uah.es}}

\affil[3]{Department of Computer Science and Engineering,
    University of Ioannina, Greece.
      \texttt{palios@cs.uoi.gr}}

\affil[4]{Departament de Matem\`{a}tiques, Universitat
    Polit\`{e}cnica de Catalunya, Spain.
      \texttt{carlos.seara@upc.edu}
}

\affil[5]{Instituto de Matem\'{a}ticas, Universidad Nacional
    Aut\'{o}noma de M\'{e}xico.
      \texttt{urrutia@matem.unam.mx}
}


\date{}

\maketitle

  \begin{abstract}
    We study the problem of rotating a simple polygon to contain the
    maximum number of elements from a given point set in the plane.~We consider
    variations of this problem where the rotation center is a given
    point or lies on a line segment, a line, or a polygonal chain.~We
    also solve an extension to 3D where we rotate a polyhedron around
    a given point to contain the maximum number of elements from a
    set of points in the space.
  \end{abstract}

\textbf{Keywords}:
Points covering, rotation, geometric optimization, polygon, polyhedron.



\section{Introduction}

Given a simple polygon~$P$ on the plane, the \emph{Polygon Placement Problem} consists in finding a function~$\uptau$,
usually
consisting of the composition of a rotation and a translation, such that $\uptau(P)$ satisfies some
geometric constraints.~In the literature, $\uptau(P)$ is known as a
\emph{placement} of $P$.~The oldest problem of this family
who, given two polygons $P$~and~$Q$, explored the problem of finding,
if it exists, a
%
%
%
placement of~$P$ that contains $Q$.~The most recent contribution to
these problems, in 2014,
%
%
can be found in \cite{barequet_2014} (see Section~1.4 there for a
summary of previous work).~Among other results, for a point set~$S$
and a simple polygon~$P$, they show how to compute a
%
placement of~$P$ that contains as many points of~$S$ as possible.~If
$n$~and~$m$ are the sizes of~$S$ and~$P$ respectively, their algorithm
runs in $O(n^3 m^3 \log (nm))$ time and $O(nm)$ space.

Although translation-only problems have also been
considered~\cite{agarwal_2002,barequet_1997}, surprisingly enough there are no previous results with~$\uptau$ being only a
rotation.~It is important to note that existing results with $\uptau$ being a composition of a rotation, a translation, and even a scaling, cannot be adapted to solve the rotation-only problem considered here: All those previous results reduce the search space complexity
by considering only placements where a constant number of points from
$S$ lie on the boundary of $P$ (see for example references
\cite{barequet_2014}~and~\cite{dickerson_1998} for algorithms based
respectively, on two-point and one-point placements).~Rotation-only
adaptations of these results would not allow the rotation center to be
fixed or restricted to lie on a given curve and therefore, cannot be
applied to the problems we deal with in this paper.~This is why the following \emph{Maximum Cover under Rotation (MCR)}
problems are considered in this paper:

\begin{problem}[Fixed MCR]\label{FixedMCR}
  \label{pro:intro:fixed_mcr}
  Given a point~$\rp$, a polygon $P$, and a point set~$S$ in the
  plane, compute an angle~$\theta\in [0,2\pi)$ such that, after
  clockwise rotating $P$ around $\rp$ by~$\theta$, the number of
  points of~$S$ contained in~$P$ is maximized.
\end{problem}

\begin{problem}[Segment-restricted MCR]
  \label{pro:intro:segment_mcr}
  Given a line segment~$\ell$, a polygon~$P$, and a point set~$S$ in
  the plane, find a point~$\rp$ on~$\ell$ and an
  angle~$\theta\in [0,2\pi)$ such that, after a clockwise rotating of
  $P$ around $\rp$ by~$\theta$, the number of points of~$S$ contained
  in~$P$ is maximized.
\end{problem}

In addition, we complete the scene opening a path towards the study of these problems in~3D, by presenting a three-dimensional version of Problem~\ref{FixedMCR}:

\begin{problem}[3D Fixed MCR]
  \label{pro:intro:3Dfixed_mcr}
  Given a point~$\rp$, a polyhedron~$P$, and a point set~$S$
  in~$\mathbb{R}^3$, compute the azimuth and
  altitude~$(\theta,\varphi)\in[0,2\pi]\times [-\pi,\pi]$ giving the
  direction in the unit sphere such that, after rotating a polyhedron
  $P$ by taking the $z$-axis to that direction, the number of points
  of~$S$ contained in~$P$ is maximized.
\end{problem}

Applications of polygon placement problems include global localization
of mobile robots, pattern matching, and geometric tolerance; see the
references in~\cite{barequet_2014}.~Rotation-only problems arise, e.g., in
robot localization using a rotating camera~\cite{robots_92}, with
applications to quality control of objects manufactured around an
axis~\cite{tolerance_1997}.

We first show that Problem~\ref{FixedMCR} is 3SUM-hard,
i.e., 
solving it in subquadratic time would imply an affirmative answer to
the open question of whether
a subquadratic time algorithm for 3SUM
exists, which is unlikely~\cite{gajentaan_1995}.~Then, we present two algorithms to solve
Problem~\ref{FixedMCR}: The first one requires $O(nm\log (nm))$ time and $O(nm)$ space,
for $n$~and~$m$ being the sizes of~$S$ and~$P$, respectively.~The
second one takes $O((n+k) \log n + m \log m)$ time and $O(n+m+k)$
space, for~$k$ in~$O(nm)$ being the number of certain events.
%
%
We also describe an algorithm that solves
Problem~\ref{pro:intro:segment_mcr} in $O(n^2 m^2\log(nm))$ time and
$O(n^2 m^2)$ space.~This algorithm can be easily extended to solve
variations of Problem~\ref{pro:intro:segment_mcr} where~$\rp$ lies on
a line or a polygonal chain.~Furthermore, our techniques for
Problem~\ref{pro:intro:fixed_mcr} can be extended to~3D to solve
Problem~\ref{pro:intro:3Dfixed_mcr} within the same time and space
complexities as Problem~\ref{pro:intro:segment_mcr}.

\section{Fixed MCR (Problem~\ref{pro:intro:fixed_mcr})}
\label{sec:fixed_mcr}

Given a point $\rp$ on the plane and a point $p \in S$, let $C_p(\rp)$
be the circle with center~$\rp$ and radius~$|\overline{\rp p}|$.~If we
rotate~$S$ in the counterclockwise direction around~$\rp$, $C_p(\rp)$
is the curve described by $p$ during a $2\pi$ rotation of $S$
around~$\rp$.~The endpoints of the circular arcs resulting from
intersecting $P$ and $C_p(\rp)$ determine the rotation angles where
$p$ enters (\emph{in-event}) and leaves (\emph{out-event})
the polygon~$P$.~In the worst case, the number of such events per element
of~$S$ is $O(m)$, see Figure~\ref{fig:fixed_mcr:comb}.~If we consider all
the points in $S$ we could get $O(nm)$ events.

\begin{figure}[ht]
  \centering{}
  \includegraphics{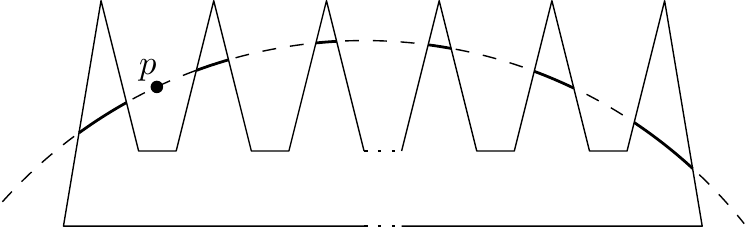}
  \caption{A comb-shaped simple polygon can generate $\Omega(m)$ in-
    and out-events per point in~$S$.}
  \label{fig:fixed_mcr:comb}
\end{figure}

\subsection{A 3SUM-hard reduction}
\label{sec:3sum-hard}

We show next that Problem~\ref{pro:intro:fixed_mcr} is 3SUM-hard, by a
reduction from the \emph{Segments Containing Points Problem} that was
proved to be 3SUM-hard in~\cite{barequet_2001}:~Given a set $A$ of $n$
real numbers and a set $B$ of $m=O(n)$ pairwise-disjoint intervals on
the real line, is there a real number $u$ such that
$A + u \subseteq B$?

\begin{theorem}\label{thm:3sum-hard}
  The Fixed MCR problem is 3SUM-hard.
\end{theorem}

\begin{proof}
  Let $I$ be an interval of the real line that contains the set $A$ of
  points, and the set $B$ of intervals of an instance of the Segments
  Containing Points Problem.~Wrap $I$ on a circle $C$ whose perimeter
  has length at least twice the length of $I$.~This effectively maps
  the points in~$A$ and the intervals in~$B$ into a set~$A'$ of points
  and a set~$B'$ of intervals on $C$.

  Clearly, finding a translation (if it exists) of the elements of $A$
  such that $A + u \subseteq B$, is equivalent to finding a rotation
  of the set of points $A'$ around the center of $C$ such that all of
  the elements of $A'$ are mapped to points contained in the intervals
  of $B'$.
  \begin{figure}[ht]
    \centering
    \subcaptionbox{\label{fig:3sum-hard:mapping:1}}
    {\includegraphics{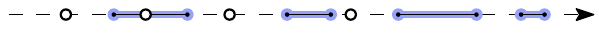}}
    \\[1.5em]
    \subcaptionbox{\label{fig:3sum-hard:mapping:2}}
    {\includegraphics{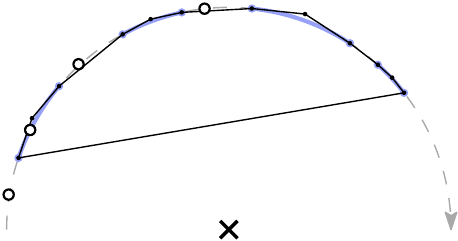}}
    \caption{Wrapping $I$ from \subref{fig:3sum-hard:mapping:1} the
      real line to \subref{fig:3sum-hard:mapping:2} a circle
      $C$.~Intervals forming $B$ and $B'$ are highlighted with
      blue.~Elements of $A$ and $A'$ are represented by white
      points.~Additional vertices forming the polygon are the
      intersection points between the tangents to $C$ at the endpoints
      of each interval in $B^{\prime}$.}
    \label{fig:3sum-hard:mapping}
  \end{figure}
  To finish our reduction, construct a polygon as shown in
  Figure~\ref{fig:3sum-hard:mapping}.
\end{proof}

\subsection{An $O(nm\log(nm))$ algorithm}\label{sec:nmlognm}

Here we present an $O(nm\log(nm))$
algorithm for
Problem~\ref{pro:intro:fixed_mcr} (note that, by Theorem~\ref{thm:3sum-hard}, this complexity is close to be optimal):

\begin{enumerate}

\item \label{enum:nmlognm:1} \textbf{Intersect rotation
    circles.}~Given a fixed point $\rp$, compute the intersection
  points of $C_{p_j}(\rp)$ and $P$, for all $p_j \in S$.~Each of these
  points determines an angle of rotation of $p_j$ around $\rp$ when
  $p_j$ enters or leaves~$P$, see
  Figure~\ref{fig:fixed_pc:in-out_events}.~These angles, in turn,
  determine a set of intervals
  $\mathcal I_j =\{I_{j,1}, \ldots , I_{j, m_j}\}$ whose endpoints
  correspond to the rotation angles in which $p_j$ enters or
  leaves~$P$ and, hence, specify the rotation angles on the unit
  circle for which $p_j$ belongs to~$P$, see again
  Figure~\ref{fig:fixed_pc:in-out_events}.~Let
  $\mathcal I = \mathcal I_1 \cup \cdots \cup \mathcal I_n$.~The set
  of endpoints of the intervals in $\mathcal I$ can be sorted in
  $O(mn \log (mn))$ time.


  \begin{figure}[ht]
    \centering
    \includegraphics{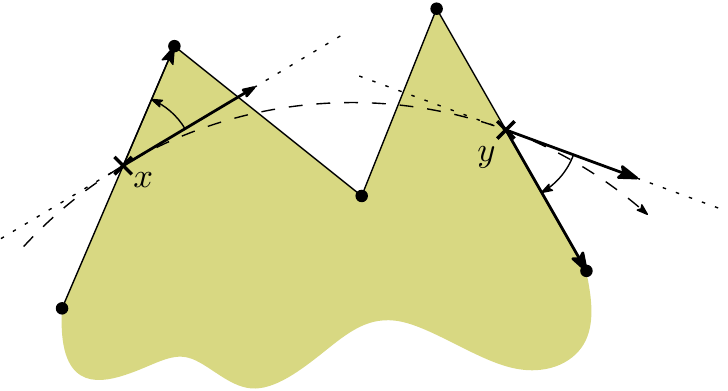}
    \caption{An in-event at $x$ (left turn), and an out-event at $y$
      (right turn).}
    \label{fig:fixed_pc:in-out_events}
  \end{figure}

\item \label{enum:nmlognm:4} \textbf{Compute the angle of maximum
    coverage.} Using standard techniques, we can now perform a sweep
  on the set $\mathcal I = \mathcal I_1 \cup \cdots \cup \mathcal I_n$
  as depicted in Figure~\ref{fig:nmlognm:table}.
  \begin{figure}[ht]
    \centering
    \includegraphics{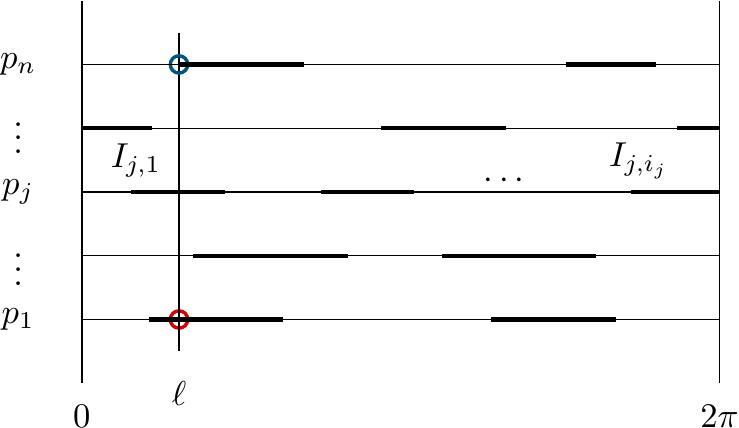}
    \caption{The events sequence and the sweeping line at angle
      $\theta$.~Highlighted with a red circle, the intersection of
      line $\ell$ with an interval corresponding to $p_1$ (where $p_1$ is
      inside $P$).~Highlighted with a blue circle, the intersection of
      line $\ell$ with one of the endpoints of an interval
      corresponding to $p_n$ (an in-event).}
    \label{fig:nmlognm:table}
  \end{figure}
  During the sweeping process, we maintain the number of points of~$S$
  lying in~$P$.~If an in-event or an out-event occurs, that number is
  increased or decreased by one, respectively.~At the end of the
  sweeping process, we report the angular interval(s) where the number
  is maximized.
\end{enumerate}

Since the complexity of our algorithm is dominated by
Step~\ref{enum:nmlognm:1}, which takes $O(nm\log (nm))$ time, we
conclude the following result.

\begin{theorem}\label{thm:nmlognm}
  The Fixed MCR problem can be solved in $O(nm\log(nm))$ time and
  $O(nm)$ space.
\end{theorem}

\subsection{An output-sensitive algorithm}
\label{sec:sensitive}

We now show that, performing a plane sweep using a \emph{sweeping
  circle} centered at~$\rp$ whose diameter increases continuously, it
is possible to intersect $P$ and the set of rotation circles in a more
efficient way.~The idea is to maintain a list of the edges
intersecting the \emph{sweeping-circle}, ordered by appearance along
the sweeping-circle.~Using the same technique shown in
Figure~\ref{fig:fixed_pc:in-out_events}, the edges are labeled as
defining in- or out- events.~The algorithm is outlined next.

\begin{enumerate}

\item \label{enum:sensitive:nm:0} \textbf{Normalize $P$}.~In the
  following steps, we consider~$P$ to have no edges intersecting any
  circle centered at $\rp$ more than once.~This can be guaranteed by
  performing a preprocessing step on $P$: For every edge~$e = uv$
  of~$P$, let $p_e$ be the intersection point between the line $\ell$
  containing $e$ and the line perpendicular to $\ell$ passing through
  $r$.~If $p_e$ belongs to the relative interior of~$e$, subdivide
  this edge into the edges $u p_e$~and~$p_e v$.~In the worst case,
  each edge of~$P$ gets subdivided into two parts. See
  Figure~\ref{fig:sensitive:mn:split_edge}.


  \begin{figure}[ht]
    \centering
    \includegraphics{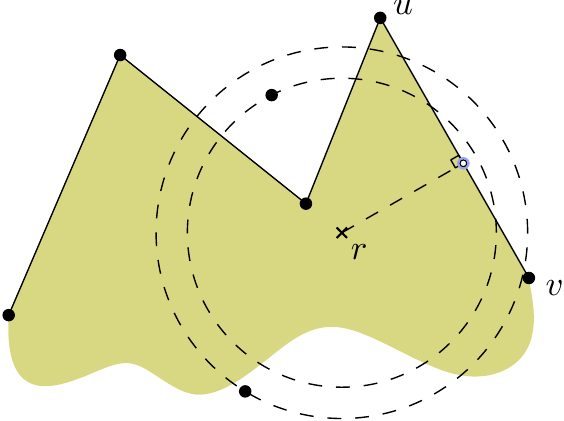}
    \caption{Splitting an edge of $P$.}
    \label{fig:sensitive:mn:split_edge}
  \end{figure}

\item \label{enum:sensitive:nm:1} \textbf{Process a vertex of $P$.}
  Sort first the vertices of $P$ and $S$ according to their distance
  from $\rp$.  This is the order in which an expanding sweeping circle
  centered at $\rp$ will reach them.

  \vspace{0.5em}
  As the sweeping-circle increases in size, we stop at each vertex
  $p_j$ of~$P$.~Each time this happens, the number of intersections of
  $C_{p_j}(\rp)$ with the boundary of~$P$ will increase or decrease by
  two.~We can maintain and update the ordered list of edges
  intersected by $C_{p_j}(\rp)$, using a red-black tree, in
  logarithmic time.~This enables us to calculate the intersections of
  $C_{p_j}(\rp)$ in time proportional to their number.~It suffices to
  walk along the ordered list of edges intersected by the
  sweeping-circle.~Each time the sweeping circle reaches an element of
  $S$, the number and order of intersections of the sweeping circle
  with the edges of $P$ remains unchanged.~However, since the points
  of intersection change, we need to recalculate them each time we
  reach a point of $P$ or $S$.

\item \label{enum:sensitive:nm:2} \textbf{Compute the intervals
    sequence for each element of $S$.} We can now compute, within the
  same time complexity, the intervals in which $C_{p_j}(\rp)$
  intersects the interior of $P$. Note that these intervals are not
  the elements of $\mathcal I_j$, they have to be rotated according to
  the position of $p_j$ with respect to $\rp$.

\item \textbf{Construct the events
    sequence.} \label{enum:sensitive:nm:3} Since for each point $p_j$
  in $S$ we have computed the corresponding sequence of sorted
  intervals $\mathcal I_j$, all we need to do is to merge these (at
  most $n$) sequences into a complete sequence of events.
\end{enumerate}

The normalization process takes $O(m)$ time.~Sorting the points in $S$
and the vertices of $P$ by distance from $r$ takes $O(n \log n)$ and
$O(m \log m)$ time, respectively. The ordered list of edges
intersecting the sweep-line is maintained in an $O(m)$-size red-black
tree, so we can process all the vertices of~$P$ in $O(m \log m)$
time.~On the other hand, processing all the points in $S$ takes $O(k)$
time, where~$k$ denotes the total number of in- and out-events in a
Fixed MCR problem.~Finally, merging the~$O(n)$ sequences of sorted
intervals takes $O(k \log n)$ time. We then sweep the merged list of
$\mathcal I_1 \cup \cdots \cup \mathcal I_n$ in $O(k)$ time to obtain
a solution to our problem.~The total time complexity
is~$O(n \log n + m \log m + k \log n)$.~The space complexity is
$O(n + m + k)$.~We have thus proved:

\begin{theorem}\label{lem:sensitive:nm}
  The Fixed MCR problem can be solved in $O((n+k) \log n + m \log m)$
  time and $O(n + m + k)$ space.
\end{theorem}

\section{Segment-restricted MCR (Problem~\ref{pro:intro:segment_mcr})}
\label{sec:segment_restricted_mcr}

Our approach to solve Problem~\ref{pro:intro:segment_mcr} is to
characterize, for each $p$ in $S$, the intersection between the
polygon~$P$ and the rotation circle~$C_p(\rp)$ while the center $\rp$
of $C_p(\rp)$ moves along a line segment $\ell = \overline{ab}$
from~$a$ to~$b$. For simplicity, we assume that $a$ lies on the
origin~$(0,0)$ and $b$ on the positive $x$-axis.  For each edge
$e = \overline{uv}$ of $P$, we parameterize the intersection
between~$C_p(\rp)$ and~$e$ using a function~$\omega = f(x)$, where $x$
is the $x$-coordinate of $\rp$ (ranging from $0$ to the
$x$-coordinate~$b.x$ of $b$) and $\omega$ is the counterclockwise
angle swept by the ray~$\overrightarrow{\rp p}$ until it coincides
with the ray emanating from $\rp$ and passing through the current
point of intersection~$q$ of $C_p(\rp)$ and~$e$ (assume for the moment
that there exists exactly one such point of intersection). See
Figure~\ref{fig:restr_mcr:fig9}.

Leaving the details for
%
%
Section~\ref{appendix}, we obtain the following expression of~$\omega$
as a function of~$x$:

\begin{equation}\label{eq:final_eq_omega}
  \omega \ = \  \arccos \left(
    \frac{\gamma(x) \pm \sqrt{\delta(x)}}{\epsilon(x)}
  \right),
\end{equation}
where $\gamma(x)$, $\delta(x)$, and $\epsilon(x)$ are polynomials of
degrees $2$, $4$, and $2$, respectively.~The motion of~$\rp$
along~$\ell$ thus corresponds to a set of points~$(x,\omega)$ for
which~$p$ hits the boundary of $P$.~For each point~$p \in S$, these
points form $O(m)$ curves bounding a collection of simple regions in
the $x$-$\omega$ plane; each point~$(x,\omega)$ of any such region
corresponds to a rotation of $p$, by a counterclockwise angle of
size~$\omega$ with respect to a rotation center at $(x,0)$, for which
$p$ belongs to $P$.~Note that each pair of such regions have disjoint
interiors, whereas their boundaries may intersect at most at a common
vertex due to the simplicity of $P$.

\subsection{Subdividing the Edges of the Polygon}

We mentioned earlier that, for convenience, we subdivide the edges of
the polygon~$P$ about their points of intersection (if any) with the
$x$-axis; so, in the following, we assume that each edge has no points
on either side of the $x$-axis.~We further subdivide the edges in
order to simplify the computation of the angle~$\omega$ in terms of
the $x$-coordinate of the rotation center~$\rp$ as it moves along the
segment~$\overline{a b}$.

\bigskip\noindent
\textbf{Theoretical Framework.}  \quad Let us consider that we process
the point~$p \in S$, and
denote by $D_p(\rp)$ the closed disk bounded by $C_p(\rp)$, where
$\rp$ is a point in $\overline{a b}$.~In Figure~\ref{fig:fig_br1}, $p$
is taken to lie above the $x$-axis where either $a.x \le p.x \le b.x$
(top figure) or $b.x < p.x$ (bottom figure).~The cases where
$p.x < a.x$ or where $p$ lies below the $x$-axis are symmetric,
whereas the case where $p$ lies on the $x$-axis is similar (see
figures~\ref{fig:fig_unique_angle1}~and~\ref{fig:fig_unique_angle2}).~Moreover,
let $p'$ be the mirror image of $p$ with respect to the $x$-axis;
clearly, $p'$ coincides with $p$ if $p$ lies on the $x$-axis.
Finally, let $H_p^L$ ($H_p^R$, resp.) be the open halfplane to the
left (right, resp.) of the line perpendicular to the $x$-axis that
passes through~$p$.

Then, it is useful to observe the following properties.

\begin{lemma} \label{lemma:circle_union} Let $p$ be a point, and let
  $H_p^L$, $H_p^R$, $C_p(r)$, and $D_p(r)$, for
  $r \in \overline{a b}$, be as defined above.
  \begin{itemize}
  \item[(i)] Consider any two points $r, r' \in \overline{a b}$ with
    $r \ne r'$.  If the point~$p$ lies on the $x$-axis, then the
    circles $C_p(r), C_p(r')$ intersect only at $p$.  If the point~$p$
    does not lie on the $x$-axis, the circles $C_p(r), C_p(r')$
    intersect at $p$ and at $p$'s mirror image~$p'$ about the
    $x$-axis, and the line segment~$\overline{p p'}$ belongs to both
    $D_p(r), D_p(r')$.
  \item[(ii)]
    \begin{itemize}
    \item[$\triangleright$\,] For every point~$s$ in the interior of
      $H_p^L \cap D_p(r)$, there exists a unique circle centered on
      the $x$-axis that passes from $p$ and $s$ and its center lies to
      the right of $r$;
    \item[$\triangleright$\,] for every point~$t$ in $H_p^L - D_p(r)$,
      there exists a unique circle centered on the $x$-axis that
      passes from $p$ and $t$ and its center lies to the left of $r$.
    \end{itemize}
    Symmetrically,
    \begin{itemize}
    \item[$\triangleright$\,] for every point~$s'$ in the interior of
      $H_p^R \cap D_p(r)$, there exists a unique circle centered on
      the $x$-axis that passes from $p$ and $s'$ and its center lies
      to the left of $r$;
    \item[$\triangleright$\,] for every point~$t'$ in
      $H_p^R - D_p(r)$, there exists a unique circle centered on the
      $x$-axis that passes from $p$ and $t'$ and its center lies to
      the right of $r$.
    \end{itemize}
  \end{itemize}
\end{lemma}

\begin{proof}~
  \paragraph{(i)} From the definition of the circles $C_p(r)$ for all
  $r \in \overline{a b}$, $p$ belongs to each such circle.

  Next, assume that $p$ lies on the $x$-axis and suppose for
  contradiction that two circles $C_p(r), C_p(r')$ with $r \ne r'$
  intersect at a point~$p' \ne p$ as well.  Then, both $r, r'$ would
  belong to the perpendicular bisector of the line
  segment~$\overline{p p'}$; thus, the perpendicular bisector should
  coincide with the $x$-axis.  Then, since $p$ lies on the $x$-axis,
  $p'$ would coincide with $p$, in contradiction to the assumption
  that $p' \ne p$.  Therefore, if $p$ lies on the $x$-axis, any two
  circles $C_p(r), C_p(r')$ with $r \ne r'$ intersect only at $p$.

  Now, assume that $p$ does not lie on the $x$-axis.  Then, since $p'$
  is the mirror image of $p$ with respect to the $x$-axis, the
  $x$-axis is the perpendicular bisector of the
  segment~$\overline{p p'}$. Thus, $p'$ belongs to all the circles
  centered on the $x$-axis that pass from $p$.  The fact that
  $\overline{p p'}$ belongs to each of the disks~$D_p(\rp)$, for all
  $\rp \in \overline{a b}$, follows from the fact that each
  disk~$D_p(\rp)$ is a convex set containing $p$ and $p'$.

  \begin{figure}[t]
    \begin{center}
      \includegraphics[height=4.8cm]{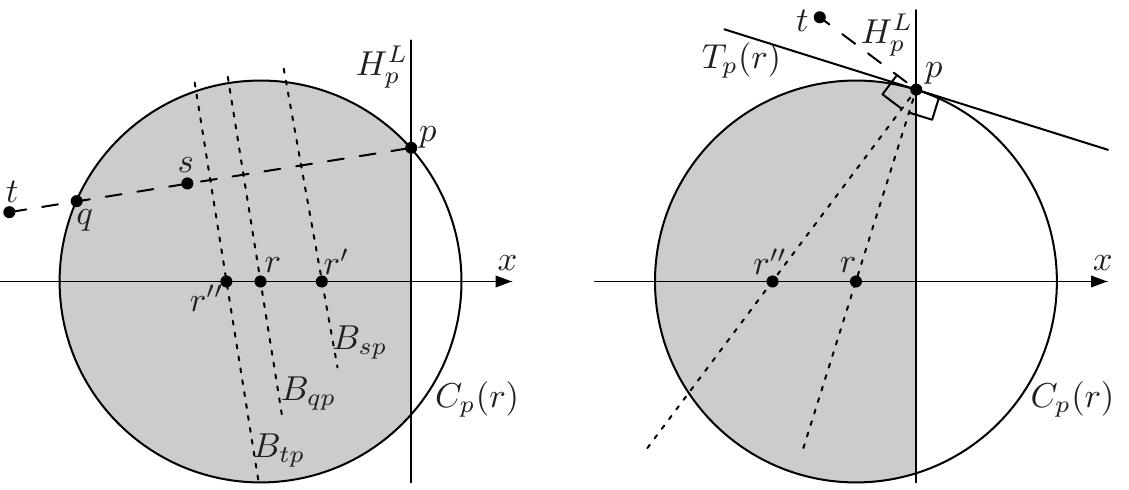}
      \caption{For the proof of Lemma~\ref{lemma:circle_union}.  \
        (left)~The perpendicular bisectors $B_{sp}$, $B_{qp}$,
        $B_{tp}$ intersect the $x$-axis at points $r', r, r''$,
        respectively.  \ (right)~The lines through $p$ that are
        perpendicular to the tangent at $p$ and to $\overline{t p}$
        intersect the $x$-axis at points $r, r''$, respectively.}
      \label{fig:fig_lemma1}
    \end{center}
  \end{figure}

  \paragraph{(ii)} Let $q$ be the point of intersection of $C_p(r)$
  with the line~$L$ through $p$ and~$s$; see
  Figure~\ref{fig:fig_lemma1}(left).~The line~$L$ is well defined
  since $s \ne p$.  In fact, $s.x < p.x$ (because $s$ belongs to
  $H_p^L$), and thus $L$ is not perpendicular to the $x$-axis, which
  implies that the perpendicular bisector~$B_{qp}$ of the line
  segment~$\overline{q p}$ intersects the $x$-axis at a single point;
  this point of intersection is precisely the center~$r$ of $C_p(r)$.
  Since the perpendicular bisector of the line
  segment~$\overline{s p}$ is parallel to $B_{qp}$ and lies to the
  right of $B_{qp}$ (because $s$ is an interior point of
  $\overline{q p}$), it intersects the $x$-axis at a single point~$r'$
  to the right of $r$; $r'$ is the center of the circle centered on
  the $x$-axis that passes from $p$ and $s$.

  Now, consider $t \in H_p^L - D_p(r)$, and let $T_p(r)$ be the open
  halfplane that is tangent to the circle~$C_p(r)$ at $p$ and contains
  $r$.  If $t \in T_p(r)$, then the line~$L$ through $p$ and $t$
  intersects $C_p(r)$ at $p$ and at another point~$q$, and
  $q \in \overline{t p}$. Then, as above, the perpendicular
  bisector~$B_{q p}$ of $\overline{q p}$ intersects the $x$-axis at
  $r$, whereas the perpendicular bisector of $\overline{t p}$ is
  parallel and to the left of $B_{q p}$ (since $q$ is an interior
  point of $\overline{t p}$), and thus intersects the $x$-axis at a
  point~$r''$ to the left of~$r$; see
  Figure~\ref{fig:fig_lemma1}(left).  It is important to observe that
  the proof so far applies no matter whether $p$ lies on the $x$-axis
  or not.

  Next, let us consider the case in which $t \not\in T_p(r)$; this
  case is not possible if $p$ lies on the $x$-axis since then
  $T_p(r) = H_p^L$.  Then, the line through $p$ perpendicular to the
  tangent to the circle~$C_p(r)$ at $p$ intersects the $x$-axis at
  $r$.  Since $t \not\in T_p(r)$, the line perpendicular to the line
  through $t$ and $p$ is not parallel to the $x$-axis and thus
  intersects the $x$-axis at a single point~$r''$.  In fact, since the
  angle $\widehat{t p r}$ of the triangle with $t, p, r$ as vertices
  is larger than $\pi/2$, $r''$ is to the left of $r$; see
  Figure~\ref{fig:fig_lemma1}(right).

  The results for points $s'$ in the interior of $H_p^R \cap D_p(r)$
  and $t' \in H_p^R - D_p(r)$ are obtained in a fashion left-to-right
  symmetric to the one we used in order to obtain the results for the
  points $s$ in the interior of $H_p^L \cap D_p(r)$ and
  $t \in H_p^L - D_p(r)$, respectively.
\end{proof}

Statement~(ii) of Lemma~\ref{lemma:circle_union} directly implies that
the union of all the circles~$C_p(\rp)$ forms precisely the closure of
the symmetric difference $D_p(a) \oplus D_p(b)$ of the disks $D_p(a)$
and $D_p(b)$ centered at $a$ and $b$, respectively (see
Figure~\ref{fig:fig_br1}); note that any point in the interior of
\[\Bigl( \bigl( D_p(a) - D_p(b) \bigr) \cap H_p^L \Bigr)
  \cup \Bigl( \bigl(D_p(b) - D_p(a) \bigr) \cap H_p^R \Bigr)\] lies on
a circle~$C_p(r)$ with $r$ in the interior of $\overline{a b}$,
whereas no other point does so.~Lemma~\ref{lemma:circle_union}(ii)
also implies the following corollary.

\begin{corollary} \label{corol:circle_union}~
  \begin{itemize}
  \item[(i)] For any $r, r' \in \overline{a b}$ with $r$ to the left
    of $r'$:
    \begin{itemize}
    \item[$\triangleright$\,]
      $\bigl( C_p(r) \cap D_p(r') \bigr) \cap H_p^L \ =\ \emptyset$ \
      \ and \ \ $D_p(r') \cap H_p^L \ \subset\ D_p(r) \cap H_p^L$;
    \item[$\triangleright$\,]
      $\bigl( C_p(r') \cap D_p(r) \bigr) \cap H_p^R \ =\ \emptyset$ \
      \ and \ \ $D_p(r) \cap H_p^R \ \subset\ D_p(r') \cap H_p^R$.
    \end{itemize}
  \item[(ii)] Suppose that a line segment~$I$ intersects a
    circle~$C_p(r)$, where $r \in \overline{a b}$, at points
    $w_1, w_2$ such that the line segment~$\overline{w_1 w_2}$ lies
    entirely in the closure of $\bigl( D_p(a) - D_p(b) \bigr)$.  Then,
    the segment~$I$ is tangent to a circle~$C_p(r')$ for some
    $r' \in \overline{a b}$ and the point of tangency belongs to
    $\overline{w_1 w_2}$.  Symmetrically, the same result holds if the
    segment~$\overline{w_1 w_2}$ lies entirely in the closure of
    $\bigl( D_p(b) - D_p(a) \bigr)$.
  \end{itemize}
\end{corollary}

\begin{proof}~
  \paragraph{(i)} We prove the propositions for the halfplane~$H_p^L$;
  the proofs for $H_p^R$ are left-to-right symmetric.

  Since $r$ is to the left of $r'$, Lemma~\ref{lemma:circle_union}(ii)
  implies that $C_p(r') \cap H_p^L$ lies in the interior of
  $D_p(r) \cap H_p^L$.~This in turn implies that
  (i)~$\bigl( C_p(r) \cap H_p^L \bigr) \cap \bigl( D_p(r') \cap H_p^L
  \bigr) = \emptyset$, i.e.,
  $\bigl( C_p(r) \cap D_p(r') \bigr) \cap H_p^L = \emptyset$, and
  (ii)~$\bigl( D_p(r') \cap H_p^L \bigr) \subset \bigl( D_p(r) \cap
  H_p^L \bigr)$ since the disk~$D_p(r')$ is bounded by $C_p(r')$ and
  since each such disk is a convex set; we have a proper subset
  relation because the points in $C_p(r) \cap H_p^L$ do not belong to
  $D_p(r') \cap H_p^L$.

  \paragraph{(ii)} Below, we prove the statement for the case that
  $\overline{w_1 w_2}$ lies entirely in the closure of
  $\bigl( D_p(a) - D_p(b) \bigr)$; the proof for the case that
  $\overline{w_1 w_2} \in \hbox{closure} \bigl( D_p(b) - D_p(a)
  \bigr)$ is left-to-right symmetric.

  Since $w_1 \ne w_2$ and
  $\overline{w_1 w_2} \in \hbox{closure} \bigl( D_p(a) - D_p(b)
  \bigr)$, then $r \ne b$; let $t \in \overline{a b}$ be a point
  infinitesimally to the right of $r$.  Then, according to
  statement~(i),
  $\bigl( C_p(r) \cap D_p(t) \bigr) \cap H_p^L = \emptyset$ and
  $\bigl( D_p(t) \cap H_p^L \bigr) \subset \bigl( D_p(r) \cap H_p^L
  \bigr)$, which together imply that
  $\bigl( D_p(t) \cap I \bigr) \subset \overline{w_1 w_2}$; note that
  at least one of $w_1, w_2$ (which belong to $C_p(r)$) belongs to
  $H_p^L$, for otherwise, either $\overline{w_1 w_2}$ degenerates to a
  single point, in contradiction to the fact that $w_1 \ne w_2$, or
  $\overline{w_1 w_2} = \overline{p p'}$ with $p \ne p'$, in
  contradiction to the fact that $\overline{w_1 w_2}$ lies entirely in
  the closure of $\bigl( D_p(a) - D_p(b) \bigr)$.  Since the rotation
  center moves continuously along $\overline{a b}$ there exists a
  point~$r' \in \overline{r b}$ such that $D_p(r') \cap I$ is a single
  point, i.e., the line segment~$I$ is tangent to the
  circle~$C_p(r')$; moreover, since
  $D_p(r') \cap I \subset \overline{w_1 w_2}$, the point of tangency
  belongs to the line segment~$\overline{w_1 w_2}$. 
\end{proof}

\begin{figure}[t]
  \begin{center}
    \includegraphics[height=7.5cm]{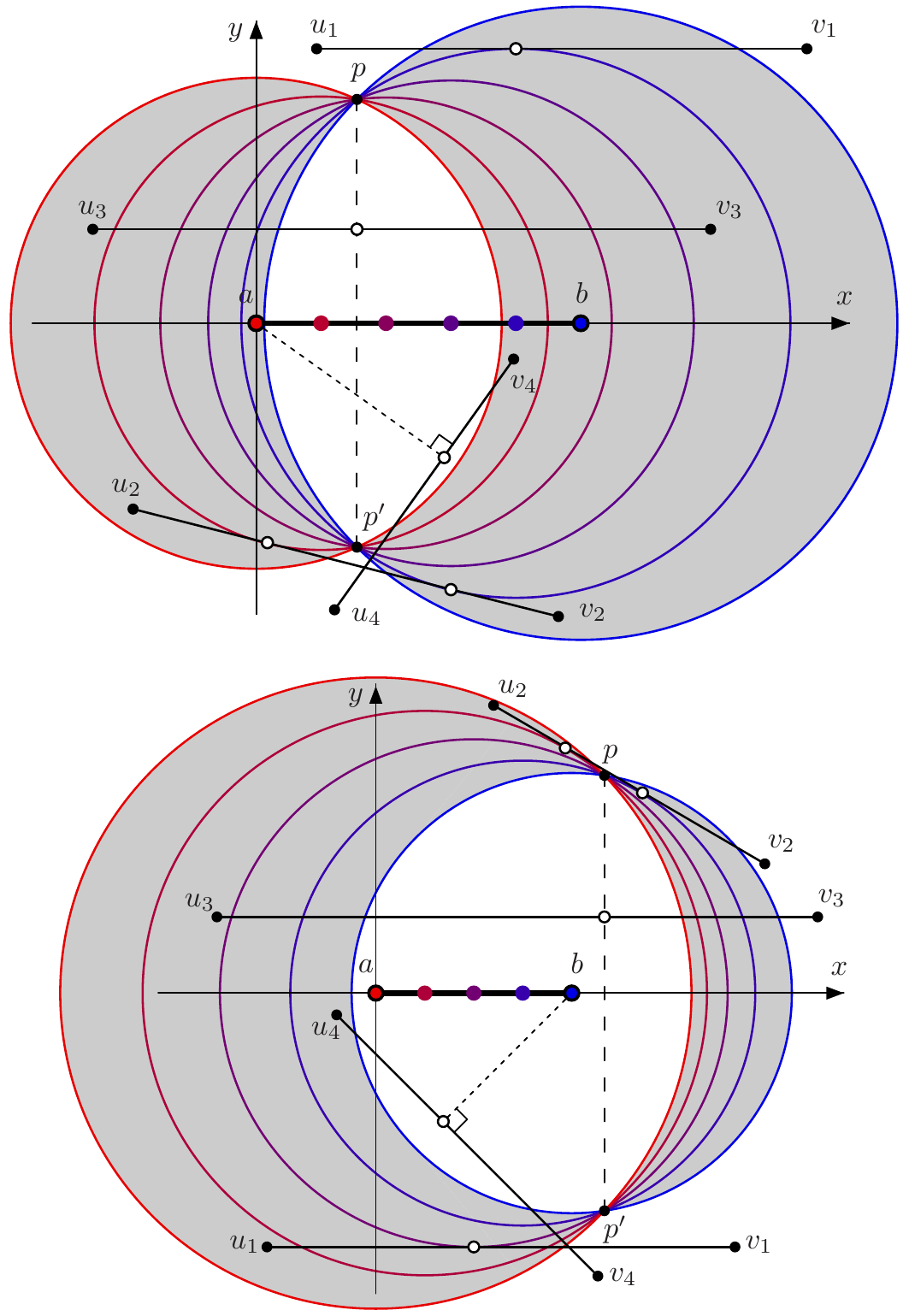}
    \caption{Subdividing the polygon edges so that each sub-edge is
      intersected at most once by each of the circles~$C_p(r)$ (white
      disks denote points of edge subdivision).}
    \label{fig:fig_br1}
  \end{center}
\end{figure}

\bigskip\noindent
\textbf{The Subdivision Procedure.} \quad Our subdivision procedure
for the polygon edges while processing point $p \in S$ works in two
phases: in Phase~1, we ensure that each circle~$C_p(\rp)$ intersects
each resulting sub-edge in at most one point; in Phase~2, we ensure
that for each sub-edge either $0 \le \omega \le \pi$ or
$\pi \le \omega \le 2 \pi$ implying that the value of $\omega$ is
uniquely determined from the value of its cosine.

\medskip\noindent
\textit{Phase~1:}
\ If an edge~$\overline{u v}$ of the polygon~$P$ does not intersect
$D_p(a) \cup D_p(b)$ or if at least one of its endpoints belongs to
$D_p(a) \cap D_p(b)$, then we need not do anything, otherwise:
\begin{itemize}
\item If $\overline{u v}$ does not intersect the interior of
  $D_p(a) \cap D_p(b)$, then $\overline{u v}$ is tangent to at most
  two of the circles $C_p(\rp)$ and we subdivide it at these points of
  tangency; see edges $\overline{u_1 v_1}$ and $\overline{u_2 v_2}$ in
  Figure~\ref{fig:fig_br1}.
\item If $\overline{u v}$ intersects the interior of
  $D_p(a) \cap D_p(b)$, then it crosses $D_p(a) \cap D_p(b)$.  If
  $\overline{u v}$ intersects the segment~$\overline{p p'}$, then we
  subdivide $\overline{u v}$ at its point of intersection with
  $\overline{p p'}$ (see edge~$\overline{u_3 v_3}$ in
  Figure~\ref{fig:fig_br1}); if not, then the points of intersection
  of $\overline{u v}$ with the boundary of $D_p(a) \cap D_p(b)$ both
  belong to either $C_p(a)$ or $C_p(b)$ (see edge~$\overline{u_4 v_4}$
  in Figure~\ref{fig:fig_br1}), in which case we subdivide
  $\overline{u v}$ at its closest point to $a$ or $b$, respectively.
\end{itemize}
It is not difficult to see that if the edge~$\overline{u v}$ has two
points of intersection with a circle~$C_p(\rp)$, these two points of
intersection end up belonging to different parts of the subdivided
edge.

After Phase~1 has been complete, we apply Phase~2 on the resulting
sub-edges.  Let $a'$ and $b'$ be points such that $a$ and $b$ are the
midpoints of segments $\overline{p a'}$ and $\overline{p b'}$,
respectively; see Figure~\ref{fig:fig_br2}.~Then, Phase~2 involves the
following subdivision steps.

\medskip\noindent
\textit{Phase~2:}
\begin{itemize}
\item If a sub-edge intersects $\overline{a' b'}$, we subdivide it at
  this point of intersection (in Figure~\ref{fig:fig_br2}, see
  sub-edges $\overline{u_1 v_1}$ and sub-edge~$\overline{u_2 v_2}$ in
  the top figure).
\item Additionally, if the sub-edge is tangent to two circles, we
  subdivide it at its point of intersection with the line through $p$
  perpendicular to the $x$-axis (see sub-edges $\overline{u_2 v_2}$ in
  Figure~\ref{fig:fig_br2}).
\end{itemize}

\begin{figure}[t]
  \begin{center}
    \includegraphics[height=7.5cm]{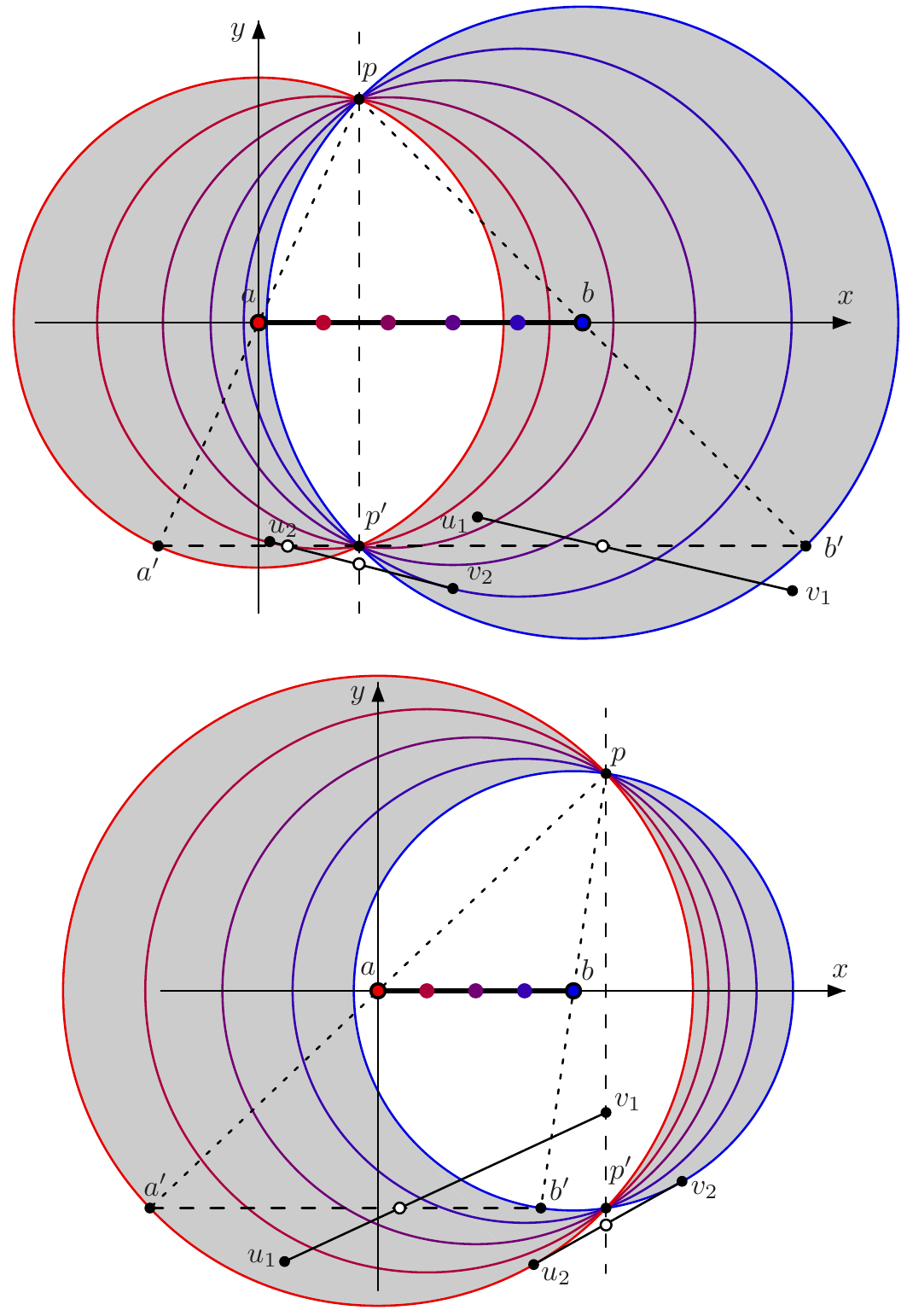}
    \caption{Further subdividing the polygon edges so that the
      angle~$\omega$ belongs either to $[0, \pi]$ or to $[\pi, 2 \pi]$
      (white disks denote points of edge subdivision).}
    \label{fig:fig_br2}
  \end{center}
\end{figure}

\bigskip
By taking into account that each of Phase~1 and Phase~2 may introduce
at most two subdivision points on a polygon edge, we conclude that
each edge ends up subdivided into at most $5$ sub-edges.

Finally, it is important to note that the above described edge
subdivision is introduced precisely for the processing of the current
point~$p \in S$ being processed; that is, for the next element of $S$,
we ignore the subdivision points introduced and start working again
with the edges of the polygon~$P$ (subdivided only about the
$x$-axis).

\bigskip\noindent
\textbf{Correctness.}
\quad Before proving Theorem~\ref{thm:edge_subdiv} which establishes
the correctness of the subdivision procedure, we show the following
useful lemma.

\begin{lemma} \label{lemma:phase2} Let $p$ be an element of the point
  set~$S$ and $p'$ be the mirror image of $p$ with respect to the
  $x$-axis.
  \begin{itemize}
  \item[(i)] If the point~$p$ is such that $0 = a.x \le p.x \le b.x$,
    then $p'$ belongs to the line segment~$\overline{a' b'}$.
  \item[(ii)] For any point~$q \in \overline{a' b'}$ such that
    $q \ne p'$, there is a point $\rp \in \overline{a b}$ for which
    $C_p(\rp)$ has the segment $\overline{q p}$ as its diameter.
  \end{itemize}
\end{lemma}

\begin{proof}
  ~\paragraph{(i)} First, assume that $p$ lies on the $x$-axis.~Then,
  $p' = p$.~The assumption $a.x \le p.x \le b.x$ implies that
  $p \in \overline{a b}$, which in turn implies that
  $\overline{a b} \subset \overline{a' b'}$; see
  Figure~\ref{fig:fig_unique_angle2}.~Thus, $p \in \overline{a' b'}$,
  i.e., $p' = p \in \overline{a' b'}$.~Now, consider the case that $p$
  does not lie on the $x$-axis.~Let $c$ be the (vertical) projection
  of $p$ onto the $x$-axis.~Since $a.x \le p.x \le b.x$,
  $c \in \overline{a b}$.~The line defined by $p, c$ (note that
  $p \ne c$) is perpendicular to the $x$-axis and let $d$ be its point
  of intersection with the line supporting $\overline{a' b'}$.  Since
  $c \in \overline{a b}$, we conclude that
  $d \in \overline{a' b'}$.~Moreover, by its construction, the line
  segment~$\overline{a' b'}$ is parallel to the $x$-axis, and since
  $|\overline{p a}| = |\overline{a a'}|$, the similarity of the
  triangles with vertices $p, a, c$ and $p, a', d$ implies that
  $|\overline{p c}| = |\overline{c d}|$.~Thus, $p' = d$ and hence
  $p' \in \overline{a' b'}$.

  \paragraph{(ii)} Assume that $p$ lies on the $x$-axis.~Let
  $q \in \overline{a' b'}$ with $q \ne p$, and suppose without loss of
  generality that $q$ is to the left of $p$ (the case where $q$ is to
  the right of $p$ is symmetric).~Then, the midpoint of
  $\overline{q p}$ lies in $\overline{a p}$ and it is the center of
  the unique circle~$C_p(r)$ passing through $q$.~Therefore, $C_p(r)$
  has $\overline{q p}$ as its diameter.

  Now assume that $p$ does not lie on the $x$-axis.~Consider any
  point~$q \in \overline{a' b'}$ with $q \ne p'$.~Let $z$ be the point
  of intersection of the line segment~$\overline{p q}$ with the
  $x$-axis ($z$ exists because $p$ and $\overline{a' b'}$, and hence
  $p$ and $q$, lie on opposite sides of the $x$-axis).~Note that
  $z \in \overline{a b}$ since $q \in \overline{a' b'}$.~Then, by the
  similarity of the triangles $\triangle p a z$ and $\triangle p a' q$
  we have that $|\overline{p z}| = |\overline{z q}|$; i.e., the
  point~$z$ is the midpoint of $\overline{p q}$.~Therefore, $z$
  belongs to the perpendicular bisector of $\overline{p q}$ and in
  fact, it is the only point of intersection of such bisector and the
  $x$-axis.~Note that, since $q \ne p'$, the line passing through $p$
  and $q$ (remember that $p \ne q$) is not perpendicular to the
  $x$-axis.~This implies that the center~$r$ of any circle~$C_p(r)$
  passing through $q$ coincides with $z$, that is, $\overline{q p}$ is
  a diameter of $C_p(r)$.
\end{proof}

\noindent
Lemma~\ref{lemma:phase2}(ii) implies that for any point~$q \ne p'$
belonging to $\overline{a' b'}$, the corresponding
angle~$\omega = \widehat{p \rp q}$ is equal to $\pi$, where
$\rp \in \overline{a b}$ is the center of the circle~$C_p(\rp)$
passing from $q$.

Now we are ready to prove Theorem~\ref{thm:edge_subdiv} which
establishes that the subdivision steps of Phases 1 and 2 achieve the
set goals.

\begin{theorem}\label{thm:edge_subdiv}~
  \begin{itemize}
  \item[(i)] After the completion of Phase~1, no resulting sub-edge
    intersects any circle~$C_p(r)$ for some $r \in \overline{a b}$ in
    more than one point.
  \item[(ii)] After the completion of Phase~2, for any two points
    $q, q'$ (lying on circles $C_p(r)$ and $C_p(r')$, respectively) of
    each resulting sub-edge, the counterclockwise angles
    $\widehat{p r q}$ and $\widehat{p r' q'}$ either both belong to
    $[0, \pi]$ or both belong to $[\pi, 2 \pi]$.
  \end{itemize}
\end{theorem}

\begin{proof}~
  \paragraph{(i)} Suppose for contradiction that there exists a
  sub-edge~$\overline{c d}$ and a circle~$C_p(r)$ with
  $r \in \overline{a b}$ that intersect in two points $w_1$ and $w_2$.
  The point~$p$ and its mirror image~$p'$ subdivide the
  circle~$C_p(r)$ into two arcs, $A_p^L$ and $A_p^R$, the former to
  the left of the line through $p$ perpendicular to the $x$-axis and
  the latter to the right (note that if $p$ lies on the $x$-axis, one
  of these arcs degenerates into a single point).  Then, $w_1, w_2$
  should belong to the same arc; otherwise, $p$ would not lie on the
  $x$-axis and the line segment~$\overline{w_1 w_2}$ would intersect
  the line segment~$\overline{p p'}$, and thus the sub-edge~$c d$
  would have been subdivided in Phase~1 about its point of
  intersection with $\overline{p p'}$.  Suppose without loss of
  generality that $w_1, w_2$ belong to the arc~$A_p^L$.  But then, no
  matter whether the segment~$\overline{w_1 w_2}$ intersects the
  interior of $D_p(a) \cap D_p(b)$ or not, we have a contradiction.
  In the former case, the sub-edge~$c d$ would have been subdivided in
  Phase~1 about the perpendicular projection of $b$ onto $c d$; $b$'s
  projection onto $c d$ belongs to $D_p(a) \cap D_p(b)$ and thus is an
  interior point of $\overline{w_1 w_2}$.  In the latter case, the
  sub-edge~$c d$ would have been subdivided in Phase~1 about its point
  of tangency with a circle~$C_p(t)$ with $t \in \overline{a b}$; this
  point of tangency belongs to $\overline{w_1 w_2}$ as shown in
  Corollary~\ref{corol:circle_union}(ii).~Therefore, after Phase~1, no
  resulting sub-edge intersects any circle~$C_p(r)$ for some
  $r \in \overline{a b}$ in more than one point.

  \paragraph{(ii)} Suppose without loss of generality that the
  point~$p$ lies above or on the $x$-axis and it holds that
  $p.x \ge a.x$; the case where it holds that $p.x < a.x$ is
  left-to-right symmetric (the corresponding angles are equal to
  $2 \pi$ minus the corresponding angles when $p.x > b.x$), whereas
  the case where $p$ lies below the $x$-axis is top-down symmetric (in
  this case too, the corresponding angles are equal to $2 \pi$ minus
  the corresponding angles when $p$ lies above the $x$-axis).

  Let $R_1$ ($R_3$, respectively) be the subsets of points in the
  closure of the symmetric difference~$D_P(a) \oplus D_p(b)$ that are
  on or to the left of the line through $p$ that is perpendicular to
  the $x$-axis and are on or above (on or below, respectively)
  $\overline{a' b'}$; symmetrically, let $R_2$ ($R_4$, respectively)
  be the subsets of points in the closure of the symmetric
  difference~$D_P(a) \oplus D_p(b)$ that are on or to the right of the
  line through $p$ that is perpendicular to the $x$-axis and are on or
  above (on or below, respectively) $\overline{a' b'}$; see
  Figure~\ref{fig:fig_unique_angle1} and
  Figure~\ref{fig:fig_unique_angle2}.  Consider a point~$w$ lying on a
  circle~$C_p(t)$ with $t \in \overline{a b}$.  Since according to
  Lemma~\ref{lemma:phase2}(ii), for any
  point~$q \in \overline{a' b'}$, the segment~$\overline{q p}$ is a
  diameter of the circle centered on the $x$-axis and passing from
  $p, q$, if $w \in R_1$, the counterclockwise angle~$\widehat{p t w}$
  belongs to $[0, \pi]$.  Similarly, if $w \in R_2$ then
  $\widehat{p t w} \in [\pi, 2 \pi]$, if $w \in R_3$ then
  $\widehat{p t w} \in [\pi, 2 \pi]$, and if $w \in R_4$ then
  $\widehat{p t w} \in [0, \pi]$.~Since no sub-edge resulting after
  Phase~2 contains points in more than one of the regions
  $R_1, R_2, R_3, R_4$, the statement of the theorem follows.
\end{proof}

\begin{figure}[t]
  \begin{center}
    \includegraphics[height=7.0cm]{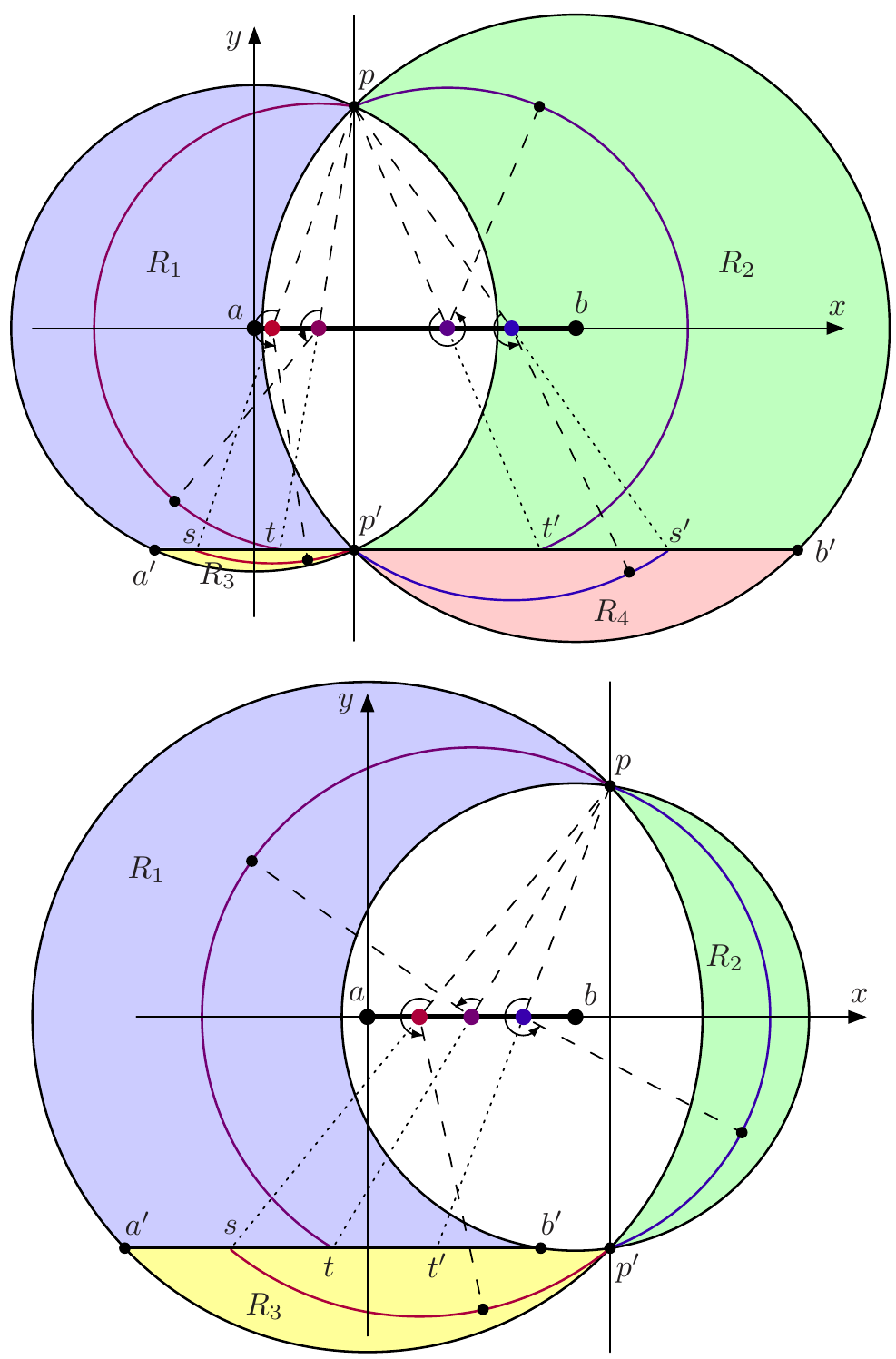}
    \caption{The partition of the closure of the symmetric difference
      $D_p(a) \oplus D_p(b)$ about the line segment~$\overline{a' b'}$
      and the line defined by $p, p'$ into regions
      $R_1, R_2, R_3, R_4$ when the point~$p$ does not lie on the
      $x$-axis.  Note that the line segments $p s, p s', p t, p t'$
      are diameters.}
    \label{fig:fig_unique_angle1}
  \end{center}
\end{figure}

\begin{figure}[t]
  \begin{center}
    \includegraphics[height=5.5cm]{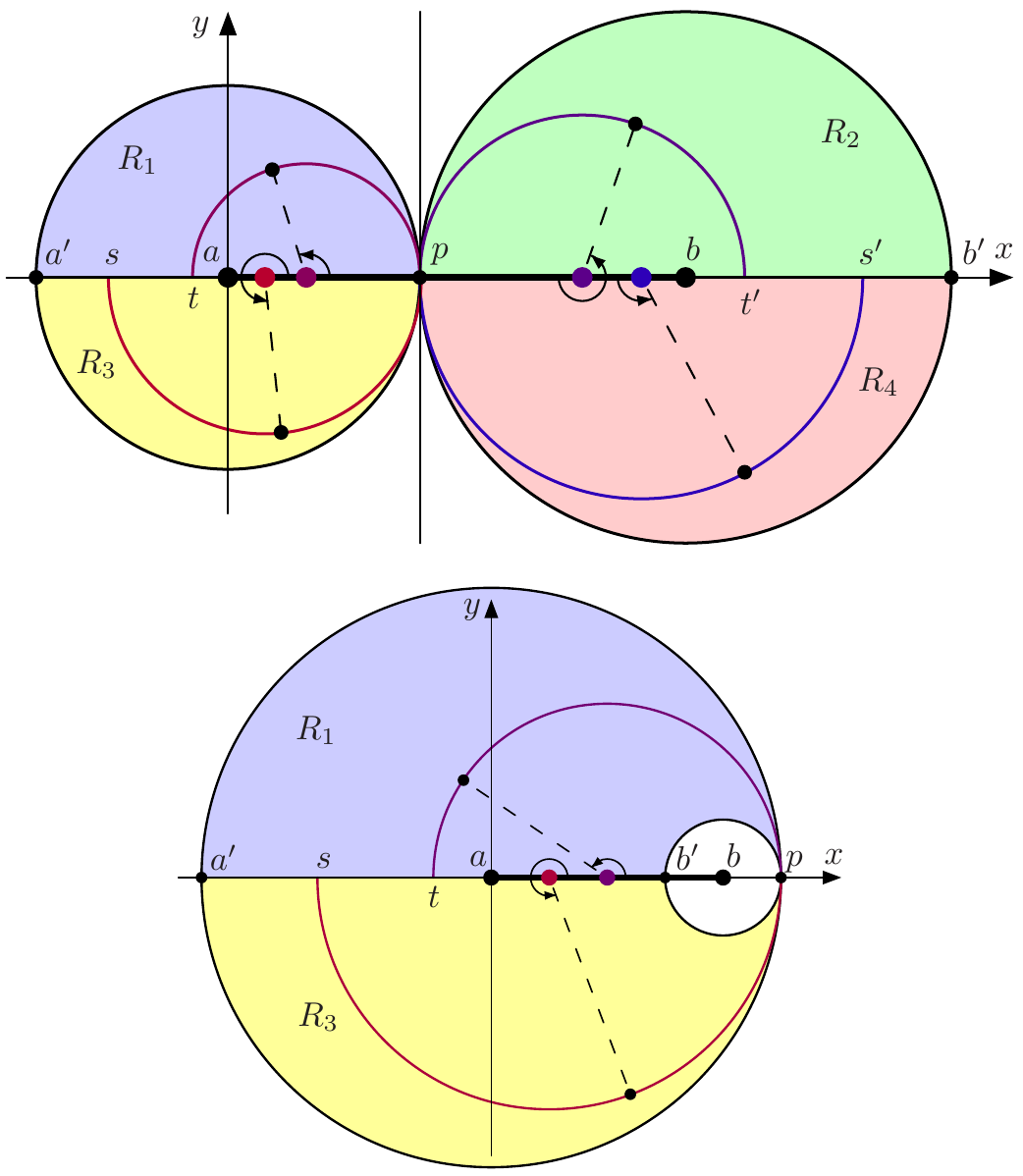}
    \caption{The partition of the closure of the symmetric difference
      $D_p(a) \oplus D_p(b)$ about the line segment~$\overline{a' b'}$
      and the line that is perpendicular to the $x$-axis at $p$ into
      regions $R_1, R_2, R_3, R_4$ when the point~$p$ lies on the
      $x$-axis.  Note that the line segments $p s, p s', p t, p t'$
      are diameters.}
    \label{fig:fig_unique_angle2}
  \end{center}
\end{figure}

\subsection{The Algorithm}

We are now ready to outline our algorithm for
Problem~\ref{pro:intro:segment_mcr}:

\begin{enumerate}

\item \label{enum:srMCR:nm:0} \textbf{Subdivide the edges of
    polygon~$P$ about the $x$-axis.}

\item \label{enum:srMCR:nm:1} \textbf{Process each point~$p \in S$.}
  For each point~$p$, we subdivide each edge of polygon~$P$ (resulting
  from the previous step) into sub-edges (see the edge subdivision
  process described earlier). Next, for each sub-edge, we compute the
  curve of the angle~$\omega$ with respect to the $x$-coordinate~$x$
  of the rotation center as it moves along $\overline{a b}$ (see
  Equation~\ref{eq:final_eq_omega}), and finally we form the regions
  bounded by these curves.

\item \label{enum:srMCR:nm:2} \textbf{Construct and traverse the
    arrangement of all the regions.}  Using standard techniques, we
  construct the arrangement of all the regions of all the elements of
  $S$.~Next, we traverse the dual graph of the resulting arrangement
  looking for a sub-region of maximum depth; any point in this
  sub-region determines a position~$(x,0)$ of $\rp$ and a rotation
  angle~$\omega$ that constitute a solution to the problem.
\end{enumerate}

\subsection{Time and Space Complexity}

Step~\ref{enum:srMCR:nm:0} clearly takes $O(m)$ time and space,
resulting into at most $2 m$ sub-edges.~The edge subdivision while
processing a point~$p \in S$ in Step~\ref{enum:srMCR:nm:1} takes
$O(m)$ time and space, producing $O(m)$ sub-edges: For each
sub-edge~$\overline{u v}$, $O(1)$ time suffices to determine whether
its endpoints belong to the disks $D_p(a)$ and $D_p(b)$, and whether
$\overline{u v}$ intersects the circles $C_p(a), C_p(b)$, the
segment~$\overline{p p'}$, or the line supporting $\overline{p p'}$,
as well as to compute any points of intersection. Moreover, the
centers of the circles $C_p(r)$, for $r \in \overline{a b}$, to which
$\overline{u v}$ is tangent are precisely the points of intersection
of the segment~$\overline{a b}$ with the parabola that is equidistant
from point~$p$ and the line supporting $\overline{u v}$.~Then,
processing $p$ yields $O(m)$ curves bounding $O(m)$ regions.~Thus,
processing all the points in $S$ in Step~\ref{enum:srMCR:nm:1} takes a
total of $O(n m)$ time and produces a set of $O(n m)$ regions bounded
by $O(n m)$ curves in the $x$-$\omega$ plane.~From
Step~\ref{eq:final_eq_omega}, we can show the following lemma:

\begin{lemma}
  Any two ($\omega$-$x$)-curves as in Equation~\ref{eq:final_eq_omega}
  have at most $32$ points of intersection.
\end{lemma}

\begin{proof}
  The idea is based on the fact that a polynomial of constant degree
  has a constant number of roots. In our case, we have a square root
  which needs to be squared in order to be removed. Let us consider
  the two ($\omega$-$x$)-curves
  \[
  \omega = \arccos \left( \frac{\gamma_1(x) \pm
    \sqrt{\delta_1(x)}}{\epsilon_1(x)} \right)
  \qquad \hbox{and} \qquad
  \omega = \arccos \left( \frac{\gamma_2(x) \pm
    \sqrt{\delta_2(x)}}{\epsilon_2(x)} \right).
  \]
  Since a point of intersection of these curves belongs to both of
  them, we have:
  \[
  \omega \ = \ \arccos \left( \frac{\gamma_1(x) \pm
    \sqrt{\delta_1(x)}}{\epsilon_1(x)} \right)
  \ = \ \arccos \left( \frac{\gamma_2(x) \pm
    \sqrt{\delta_2(x)}}{\epsilon_2(x)} \right)
  \]
  \begin{equation}
    \Longrightarrow
    \   \gamma_1(x) \, \epsilon_2(x) - \gamma_2(x) \, \epsilon_1(x)
    \  = \  \pm \Bigl( \epsilon_2(x) \, \sqrt{\delta_1(x)} -
    \epsilon_1(x) \, \sqrt{\delta_2(x)} \Bigr)
    \label{eq:eq1}
  \end{equation}
  from which, by squaring twice to get rid of the square
  roots, we get

  \begin{align}\label{eq:poly}
    & \Bigl( \gamma_1(x) \, \epsilon_2(x) - \gamma_2(x) \, \epsilon_1(x)
      \Bigr)^2
      \  = \  \Bigl( \epsilon_2(x) \, \sqrt{\delta_1(x)} -
      \epsilon_1(x) \, \sqrt{\delta_2(x)} \Bigr)^2 \nonumber \\
    & \Longrightarrow
      \   \Bigl(
      \gamma_1(x) \, \epsilon_2(x) - \gamma_2(x) \, \epsilon_1(x)
      \Bigr)^2
      - {\epsilon^2_2(x)} \, {\delta_1(x)}
      - {\epsilon^2_1(x)} \, {\delta_2(x)}
      \nonumber \\   
    & \phantom{xxxxxxxxxxxxxxxxxxxxxxxxxxxxx}
      \  =
      \  -2 \, {\epsilon_1(x) \, \epsilon_2(x)}
      \, {\sqrt{\delta_1(x) \, \delta_2(x)}}
      \nonumber \\
    & \Longrightarrow
      \     \left( \Bigl( \gamma_1(x) \, \epsilon_2(x)
      - \gamma_2(x) \, \epsilon_1(x) \Bigr)^2
      - {\epsilon^2_2(x)} \, {\delta_1(x)}
      - {\epsilon^2_1(x)} \, {\delta_2(x)} \right)^2
      \nonumber \\   
    & \phantom{xxxxxxxxxxxxxxxxxxxxxxxxxxxxx}
      \  =
      \  4 \, {\epsilon^2_1(x) \, \epsilon^2_2(x)}
      \, \delta_1(x) \, \delta_2(x).
  \end{align}
  The last equality is a polynomial of degree at most $16$ and, thus,
  it has at most~$16$ real roots for $x$ (it is important to note that
  the value of $x$ in any pair $(\omega,x)$ satisfying
  Equation~\ref{eq:eq1} satisfies the polynomial in
  Equation~\ref{eq:poly}, although the reverse does not necessarily
  hold, i.e., not every root of the polynomial satisfies
  Equation~\ref{eq:eq1}).~Thus, if we substitute the real roots of the
  polynomial in Equation~\ref{eq:poly} into
  Equation~\ref{eq:final_eq_omega}, we get at most $32$ possible
  points of intersection, due to the $\pm$ operand.
\end{proof}

Hence, the total number of intersection points of all the curves is
$O(n^2 m^2)$.~Using standard techniques, in $O(n^2 m^2 \log (nm))$
time the arrangement of all these regions can be computed, and the
dual graph of the resulting arrangement can be traversed looking for a
sub-region of maximum depth.~Any point in this sub-region determines a
position of the rotation center~$\rp$ and a rotation angle~$\omega$
that constitute a solution to the problem.~The space complexity is
$O(n^2 m^2)$. Then:

\begin{theorem}\label{lem:restr:nm}
  The Segment-restricted MCR problem 
  can be solved in\\
  $O(n^2 m^2\log (n m))$ time and $O(n^2 m^2)$ space.
\end{theorem}

Note that Problem~\ref{pro:intro:segment_mcr} can also be solved in
$O(n^2 m^2\log (n m))$ time even when the rotation center is
restricted to lie on a line~$L$: Compute the Voronoi diagram
of~$P \cup S$, and apply the algorithm we just described to a segment
of~$L$ containing all the intersection points of~$L$ and the Voronoi
edges.~Moreover, if we restrict the rotation center to lie on a
polygonal chain with $s$ line segments, we can trivially obtain the
optimal placement of~$P$ using $O(sn^2 m^2\log (n m))$ time.~In both
cases, the space complexity is $O(n^2 m^2)$.

\subsection{Equation~\ref{eq:final_eq_omega}: expressing~$w$ as a function of~$x$}
\label{appendix}

In order to simplify the exposition leading to
Equation~\ref{eq:final_eq_omega}, for each point~$s$ in the plane
other than the current rotation center~$\rp$, we define a
corresponding angle~$\vartheta_s$ with respect to $\rp$.~In
particular, let $\Hur$ be the set of points above the $x$-axis or on
the $x$-axis and to the right of $\rp$ and let $\Hdl$ be the set of
points below the $x$-axis or on the $x$-axis and to the left of $\rp$
(clearly, the sets $\Hur$ and $\Hdl$ partition
$\mathbb{R}^2 - \{\rp\}$).  Then,
\begin{itemize}
\item if $s \in \Hur$, $\vartheta_s$ is the angle swept by the
  rightward horizontal ray emanating from $\rp$ as it moves in
  counterclockwise direction around $\rp$ until it coincides with the
  ray~$\overrightarrow{\rp s}$ (see Figure~\ref{fig:fig3.1}, left);
\item if $s \in \Hdl$, $\vartheta_s$ is the angle swept by the
  leftward horizontal ray emanating from $\rp$ as it moves in
  counterclockwise direction around $\rp$ until it coincides with the
  ray~$\overrightarrow{\rp s}$ (see Figure~\ref{fig:fig3.1}, right).
\end{itemize}
(Note that for all points~$s$ on the $x$-axis, $\vartheta_s = 0$.)
From the definition of $\vartheta_s$, it follows that in all cases
\begin{equation}
  \label{eq:theta}
  0 \le \vartheta_s < \pi
\end{equation}
(we consider counterclockwise and clockwise angles being positive and
negative, respectively) and

\begin{equation}
  \label{eq:cos_sin}
  \cos \vartheta_s = \frac{s.x - \rp.x}{d(s,\rp)} \, sgn(s.y)
  \qquad
  \sin \vartheta_s = \frac{|s.y|}{d(s,\rp)} =
  \frac{s.y}{d(s,\rp)} \, sgn(s.y)
\end{equation}
where $d(s,\rp)$ denotes the distance of point~$s$ from the rotation
center~$\rp$, $p.x$ and $p.y$ are respectively the \noindent
$x-$ and $y-$coordinates of a point~$p$, and $sgn(s.y)$ is the sign of $s.y$.

\begin{figure}[t]
  \centering

  \psfrag{s}[][]{$s$}
  \psfrag{r}[][]{$\rp$}
  \psfrag{x}[][]{$x$}
  \psfrag{d.}[][]{$\vartheta_s$}

  \includegraphics[height=2.5cm]{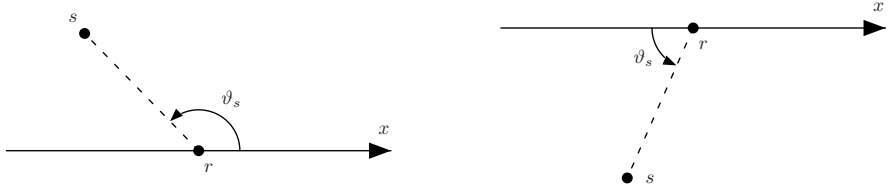}
  \caption{The definition of the angle~$\vartheta_s$ for any
    point~$s \neq \rp$.}
  \label{fig:fig3.1}
\end{figure}

Now, we distinguish two main cases:

\begin{itemize}
\item \emph{Point~$p$ and the intersection point~$q$ of the
    circle~$C_p(\rp)$ and the edge~$e = \overline{uv}$ of $P$ both
    belong to either $\Hur$ or $\Hdl$ (see
    Figure~\ref{fig:restr_mcr:fig9}(a))}: if
  $\vartheta_p \le \vartheta_q$ then
  \begin{equation}
    \label{eq:same2}
    \omega \  = \  \vartheta_q - \vartheta_p
  \end{equation}
  otherwise
  \begin{equation}
    \label{eq:same1}
    \omega \  = \  (\pi - \vartheta_p) + \pi + \vartheta_q
    \  = \  2 \pi + \vartheta_q - \vartheta_p.
  \end{equation}

\item \emph{Point~$p$ and the intersection point~$q$ of the
    circle~$C_p(\rp)$ and the edge~$e = \overline{uv}$ of $P$ do not
    both belong to either $\Hur$ or $\Hdl$ (see
    Figure~\ref{fig:restr_mcr:fig9}(b))}: in this case,
  \begin{equation}
    \label{eq:opposite}
    \omega \  = \  (\pi - \vartheta_p) + \vartheta_q
    \  = \  \pi + \vartheta_q - \vartheta_p.
  \end{equation}
\end{itemize}
It is important to observe that the definition of $\Hur$ and $\Hdl$
ensures that the above expressions for $\omega$ hold for all special
cases in which at least one of $p, q$ lies on the $x$-axis, as
summarized in the following table.
\smallskip
{
\footnotesize
\begin{center}
  \begin{tabular}{cc|c|c|c|c|}
    \cline{3-6}
    & & \multicolumn{2}{c|}{$p \in \Hur$} & \multicolumn{2}{c|}{$p \in \Hdl$}
    \\
    & & {\scriptsize $p$ on $x$-axis} & {\scriptsize $p$ above $x$-axis}
    & {\scriptsize $p$ on $x$-axis} & {\scriptsize $p$ below $x$-axis}
    \\
    & & {\small $\vartheta_p = 0$} & {\small $0 < \vartheta_p < \pi$}
    & {\small $\vartheta_p = 0$} & {\small $0 < \vartheta_p < \pi$}
    \\
    \cline{1-6}
    \multicolumn{1}{ |c }{\multirow{4}{*}{$q \in \Hur$}}
    & \multicolumn{1}{ c| }{{\scriptsize $q$ on $x$-axis}}
      & {\multirow{2}{*}{ $\omega = 0$ }} & {\multirow{2}{*}{ $\omega = 2 \pi - \vartheta_p$ }}
    & {\multirow{2}{*}{ $\omega = \pi$ }} & {\multirow{2}{*}{ $\omega = \pi - \vartheta_p$ }}
    \\
    \multicolumn{1}{ |c  }{} & \multicolumn{1}{ c| }{{\small $\vartheta_q = 0$}} & {} & {} & {} & {}
    \\ \cline{2-6}
    \multicolumn{1}{ |c  }{} & \multicolumn{1}{ c| }{{\scriptsize $q$ above $x$-axis}}
      & {\multirow{2}{*}{ $\omega = \vartheta_q$ }} & {\multirow{2}{*}{ {\small Eq.~(\ref{eq:same2}), (\ref{eq:same1})} }}
    & {\multirow{2}{*}{ $\omega = \pi + \vartheta_q$ }}
      & {\multirow{2}{*}{ {\small Eq.~(\ref{eq:opposite})} }}
    \\
    \multicolumn{1}{ |c  }{} & \multicolumn{1}{ c| }{{\small $0 < \vartheta_q < \pi$}} & {} & {} & {} & {}
    \\ \cline{1-6}
    \multicolumn{1}{ |c  }{\multirow{4}{*}{$q \in \Hdl$} } & \multicolumn{1}{ c| }{{\scriptsize $q$ on $x$-axis}}
      & {\multirow{2}{*}{ $\omega = \pi$ }} & {\multirow{2}{*}{ $\omega = \pi - \vartheta_p$ }}
    & {\multirow{2}{*}{ $\omega = 0$ }} & {\multirow{2}{*}{ $\omega = 2 \pi - \vartheta_p$ }}
    \\
    \multicolumn{1}{ |c  }{} & \multicolumn{1}{ c| }{{\small $\vartheta_q = 0$}} & {} & {} & {} & {}
    \\ \cline{2-6}
    \multicolumn{1}{ |c  }{} & \multicolumn{1}{ c| }{{\scriptsize $q$ below $x$-axis}}
      & {\multirow{2}{*}{ $\omega = \pi + \vartheta_q$ }} & {\multirow{2}{*}{ {\small Eq.~(\ref{eq:opposite})} }}
    & {\multirow{2}{*}{ $\omega = \vartheta_q$ }}
      & {\multirow{2}{*}{ {\small Eq.~(\ref{eq:same2}), (\ref{eq:same1})} }}
    \\
    \multicolumn{1}{ |c  }{} & \multicolumn{1}{ c| }{{\small $0 < \vartheta_q < \pi$}} & {} & {} & {} & {}
    \\ \cline{1-6}
  \end{tabular}
\end{center}
}
\smallskip

\begin{figure}[ht]
  \centering

  \psfrag{a}[][]{$a$}
  \psfrag{b}[][]{$b$}
  \psfrag{p}[][]{$p$}
  \psfrag{q}[][]{$q$}
  \psfrag{r}[][]{$\rp$}
  \psfrag{u}[][]{$u$}
  \psfrag{v}[][]{$v$}
  \psfrag{x}[][]{$x$}
  \psfrag{y}[][]{$y$}
  \psfrag{w}[][]{$\omega$}
  \psfrag{d.}[][]{$\vartheta_p$}
  \psfrag{d,}[][]{$\vartheta_q$}

  \subcaptionbox{\label{fig:restr_mcr:fig9:1}}
  {\includegraphics[height=4.6cm]{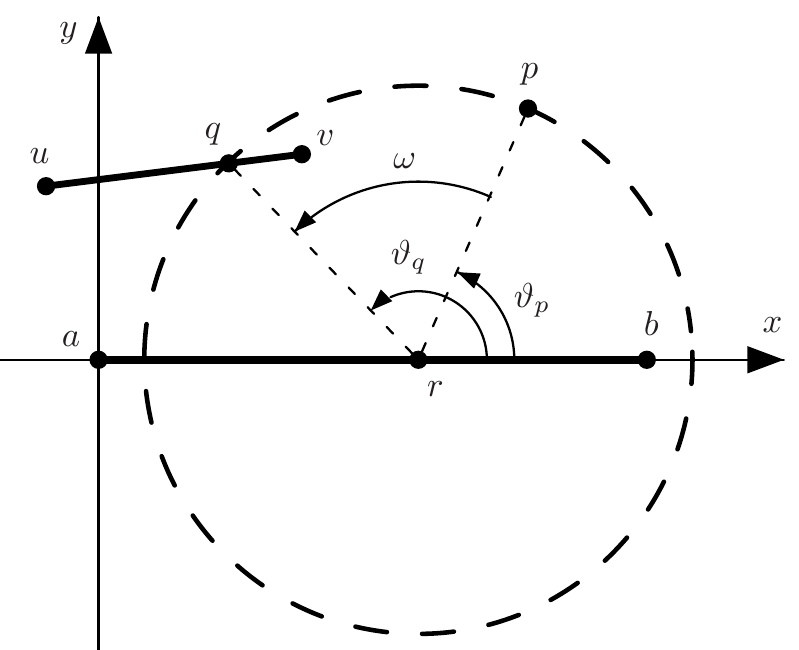} \hspace{3em}
    \includegraphics[height=4.6cm]{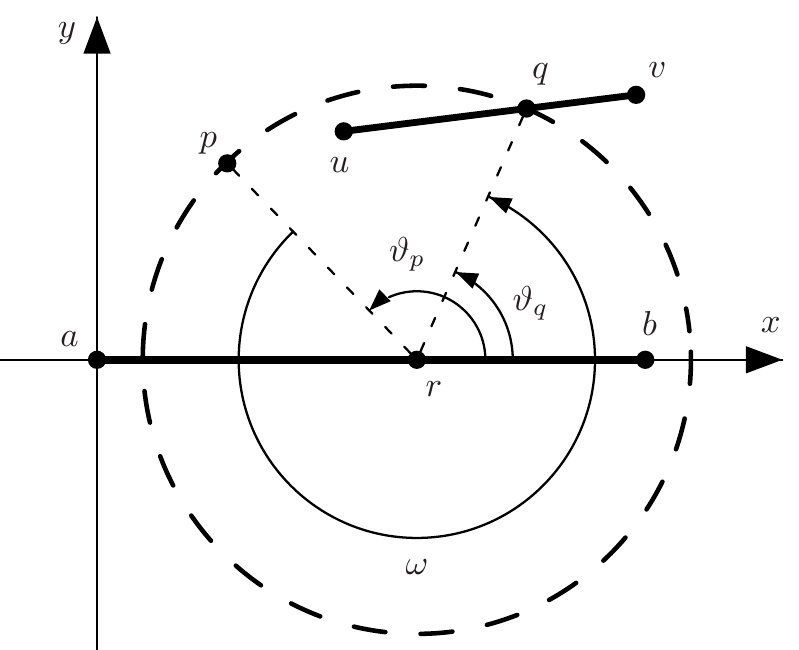}}
  \\[2em]
  \subcaptionbox{\label{fig:restr_mcr:fig9:2}}
  {\includegraphics[height=4.6cm]{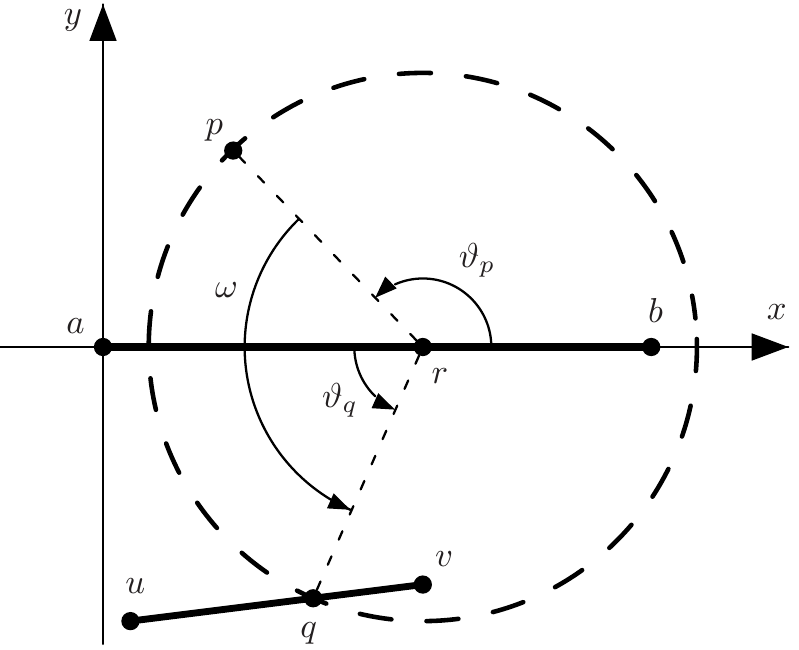}}

  \caption{
    Parameterizing the intersection between the circle~$C_p(\rp)$ and
    the edge~$\overline{uv}$ while $\rp$ moves along
    segment~$\overline{ab}$ when point~$p$ and the intersection~$q$ of
    $C_p(\rp)$ and $\overline{uv}$ are
    \subref{fig:restr_mcr:fig9:1}~in the same halfplane and
    \subref{fig:restr_mcr:fig9:2}~in opposite halfplanes (with respect
    to the $x$-axis).}
  \label{fig:restr_mcr:fig9}
\end{figure}

In all cases,
$\cos(\omega) \ = \ \cos(\vartheta_q - \vartheta_p) \ = \
\cos(\vartheta_q) \, \cos(\vartheta_p) + \sin(\vartheta_q) \,
\sin(\vartheta_p)$ which, due to Equation~\ref{eq:cos_sin} and to the
fact that $d(q,r) = d(p,r)$, implies that

\begin{align}\label{eq:eq_omega}
  \cos(\omega)
  &  \  =  \  \frac{(q.x - x) \, (p.x - x) + q.y \, p.y}
    {d^2(p,\rp)} \, sgn(q.y) \, sgn(p.y) \nonumber \\
  &  \  =  \  \frac{(q.x - x) \, (p.x - x) + q.y \, p.y}
    {(p.x - x)^2 + (p.y)^2} \, sgn(q.y) \, sgn(p.y) \nonumber \\
  &  \  =  \  \frac{x^2 - (q.x + p.x) \, x
    + q.x \, p.x + q.y \, p.y}
    {x^2 - 2 \, p.x \, x + (p.x)^2 + (p.y)^2}
    \, sgn(q.y) \, sgn(p.y).
\end{align}

For convenience, we subdivide each edge that intersects the $x$-axis
at this point of intersection so that the value of $sgn(q.y)$ is fixed
at each sub-edge no matter where $q$ is.

The coordinates $q.x, q.y$ of intersection point~$q$ can be expressed
in terms of~$x$ by taking into account that $q$ belongs to the line
supporting the edge~$\overline{u v}$ and that $\rp$ is equidistant
from $q$ and $p$.~The former implies that there exists a real
number~$\lambda$ with $0 \le \lambda \le 1$ such that the
vector~$\overrightarrow{u q}$ is $\lambda$ times the
vector~$\overrightarrow{u v}$, which yields
\begin{equation}\label{eq:q_in_uv_x}
  (q.x - u.x) \ = \ \lambda \,
  (v.x - u.x) \quad \Longleftrightarrow \quad q.x \ = \ \lambda \,
  (v.x - u.x) + u.x
\end{equation}
and
\begin{equation}\label{eq:q_in_uv_y}
  (q.y - u.y) \  = \  \lambda \, (v.y - u.y)
  \quad \Longleftrightarrow  \quad
  q.y \  = \  \lambda \, (v.y - u.y) + u.y,
\end{equation}
whereas the latter implies

\begin{align}\label{eq:equidistant}
  &  \quad  d^2(q,\rp) \  = \  d^2(p,\rp) \nonumber \\
  \Longleftrightarrow & \quad (q.x - x)^2 + (q.y)^2
                        \  = \  (p.x - x)^2 + (p.y)^2 \nonumber \\
  \Longleftrightarrow & \quad (q.x)^2 - 2 \, x \, q.x + (q.y)^2 -
                        (p.x)^2 + 2 \, x \, p.x - (p.y)^2 \ = \ 0.
\end{align}

By substituting $q.x, q.y$ from equations~\ref{eq:q_in_uv_x} and
\ref{eq:q_in_uv_y} into Equation~\ref{eq:equidistant}, we get
\begin{align*}
  & \quad \bigl[ \lambda \, (v.x - u.x) + u.x \bigr]^2
    - 2 \, x \, \bigl[ \lambda \, (v.x - u.x) + u.x \bigr]
  \\
  & \quad + \bigl[ \lambda \, (v.y - u.y) + u.y \bigr]^2
    - (p.x)^2 + 2 \, x \, p.x - (p.y)^2 \  = \  0   
  \\
  \Longleftrightarrow & \quad \lambda^2 \, \left[ (v.x - u.x)^2
                        + (v.y - u.y)^2 \right]   
  \\
  & \quad - 2 \, \lambda \, \bigl[ x \, (v.x - u.x)
    - u.x \, (v.x - u.x) - u.y \, (v.y - u.y) \bigr]
  \\
  & \quad - 2 \, x \, (u.x - p.x) + (u.x)^2 + (u.y)^2
    - (p.x)^2 - (p.y)^2 \  = \  0, 
\end{align*}

which has at most $2$ roots for $\lambda$ in terms of $x$ of the form
\begin{equation}\label{eq:lambda_equation}
  \lambda \  = \  \alpha(x) \pm \sqrt{\beta(x)},
\end{equation}
where $\alpha(x)$ and $\beta(x)$ are polynomials of degrees $1$ and
$2$, respectively.

Then, by substituting $q.x, q.y$, and $\lambda$ from equations
\ref{eq:q_in_uv_x}, \ref{eq:q_in_uv_y} and \ref{eq:lambda_equation}
respectively, into Equation~\ref{eq:eq_omega}, we get:
\begin{equation}
  \cos (\omega)
    \  = \  \frac{\gamma(x) \pm \sqrt{\delta(x)}}{\epsilon(x)}
\ \quad \Longrightarrow \ \quad
\omega \ = \  \arccos \left(
\frac{\gamma(x) \pm \sqrt{\delta(x)}}{\epsilon(x)}
\right),
\end{equation}
where $\gamma(x)$, $\delta(x)$, and $\epsilon(x)$ are polynomials of
degrees $2$, $4$, and $2$, respectively.

\section{3D Fixed MCR (Problem~\ref{pro:intro:3Dfixed_mcr})}

In this section we extend our techniques to the 3D-equivalent of
Problem~\ref{pro:intro:fixed_mcr}.
We consider a set $S$ of $n$ points in 3D, a rotation center $r$, and
a non self-intersecting polyhedron $P$ with complexity $m$, i.e., with
$m$ facets.  We identify rotations around~$r$ with points in a sphere
with center $r$.~The following shows how to extend the algorithm we
used to solve the Fixed MCR problem:

\begin{enumerate}
\item {\bf Compute the inclusion regions}. For each
    $p_j\in S$, the intersection of the sphere~$C_{p_j}(\rp)$ with center at
  $\rp$ and radius~$|\overline{\rp p_j}|$ with the polyhedron~$P$ results in a set of regions on
  the boundary of the sphere.~These regions consist of the rotated
  copies of $p_j$
  that lie in the interior of $P$.
    \begin{itemize}
    \item Regardless of $P$ being convex or not, each facet can contribute to those regions a constant number of times. Hence, the overall complexity is $O(m)$.~Moreover, notice that a region can have many
      holes, even in the case that~$P$ is convex.
    \item The sides of these regions on the sphere $C_{p_j}(\rp)$ are
      arcs of circles, since they are the intersection of the sphere
      with a planar facet of the polyhedron.~Then, these sides can be
      computed in constant time each, as the intersection of the planes
      containing the faces of the polyhedron with~$C_{p_j}(\rp)$.
    \item Thus the total time and space complexities of computing all
      the $O(nm)$ regions is $O(nm)$.
    \end{itemize}



  \item {\bf Normalize inclusion regions}. Let $R_{p_j}$ be the set of
    inclusion regions of $p_j\in S$. Consider the unit sphere $S^2$ to
    be centered at $r$ and project the regions to $S^2$. Choose a
    point $N$ in $S^2$ as reference and compute the rotation $\tau_j$
    required to send $p_j$ to $N$. Then compute $\tau_j(R_{p_j})$ to
    set the same reference for all the inclusion regions.

  \item {\bf Computing the depth of $N$}. For later use, we need to
    compute how many of the above regions contain the point $N$ (in
    its interior or boundary), what we call the \emph{depth} of
    $N$. In order to compute it, we perform point location in the
    planar subdivision on the sphere, i.e., we check whether the point
    $N$ belongs to each of the $O(nm)$ regions with a cost of
    $O(\log m)$ per region, for a total time complexity of
    $O(nm\log m)$.

  \item {\bf Stereographic projection}. We use the well-known
    stereographic projection from the point $N$, considered as the
    north pole, to the tangent plane at the antipodal south pole. The
    fact that this projection is conformal implies that circles in the
    sphere are mapped to circles in the
    plane~\cite{needham_2002}.~Therefore, the projections of the
    inclusion regions $\tau_j(R_{p_j})$ have boundaries composed by
    circular arcs.~Because any two sides (arcs of circles) of the
    regions can intersect at most two times, the arrangement
    $\mathcal{A}$ of projected regions can be computed in $O(n^2m^2)$
    time and space, since the total number of intersection points
    between arcs is $O(n^2m^2)$.~Notice that for computing the
    projected arc we proceed as follows: We compute the projection of
    the two endpoints of the arc, and also the projection of a third
    point of the arc (for example the corresponding to the midpoint of
    the arc); with these three projected points, we compute the circle
    containing the projected arc and the projected arc itself.


  \item {\bf Computing the region in $\mathcal{A}$ with largest
      depth}. To do this computation we work on the dual graph of the
    arrangement $\mathcal{A}$, just knowing that the exterior
    (unbounded) face of $\mathcal{A}$ is the face which was containing
    the point~$N$, and hence we know its depth.~Starting in this face,
    we perform a traversal of the dual graph, computing the depth of
    each region and maintaining the region with maximum depth, in a
    total $O(n^2m^2)$ time.

    Computing an interior point of the region with maximum depth, we
    compute its corresponding point in the unit sphere and then we
    know the two parameters $\theta,\varphi$ giving such direction,
    which is the solution of our problem.

\end{enumerate}

\begin{theorem}\label{MCR-3D}
  The Fixed MCR problem in 3D 
  can be
  solved in $O(n^2m^2\log (nm))$ time and $O(n^2m^2)$ space.
\end{theorem}

\section{Concluding Remarks}

We studied the problem of finding a rotation of a simple polygon that
covers the maximum number of points from a given point set.~We
described algorithms to solve the problem when the rotation center is
fixed, or lies on a line segment, a line, or a polygonal
chain.~Without much effort, our algorithms can also be applied when
the polygon has holes, and can be easily modified to solve
minimization versions of the same problems. We also solved the problem
with a fixed rotation center in 3D, leaving as open problem the
3D-analogue of Problem~\ref{pro:intro:segment_mcr}.


\section{Acknowledgements}

David Orden is supported by MINECO Projects MTM2014-54207 and MTM2017-83750-P, as well as
by H2020-MSCA-RISE project 734922 - CONNECT.~Carlos
Seara is supported by projects Gen. Cat. DGR 2017SGR1640, by MINECO
MTM2015-63791-R, and by H2020-MSCA-RISE project 734922 - CONNECT.~Jorge
Urrutia is supported in part by SEP-CONACYT
of M\'{e}xico, Proyecto 80268 and by PAPPIIT  IN102117 Programa de Apoyo a la Investigación e Innovación Tecnológica, Universidad Nacional Autónoma de México..



\bibliographystyle{plain}

\begin{thebibliography}{10}

\bibitem{agarwal_2002}
P.~K. Agarwal, T.~Hagerup, R.~Ray, M.~Sharir, M.~Smid, and E.~Welzl.
\newblock Translating a planar object to maximize point containment.
\newblock In \emph{Algorithms --- ESA 2002: 10th Annual European Symposium.
  Rome, Italy, September 17--21, 2002. Proceedings}, pages 42--53, 2002.

\bibitem{barequet_2014}
G.~Barequet and A.~Goryachev.
\newblock Offset polygon and annulus placement problems.
\newblock \emph{{Computational Geometry: Theory and Applications}}, 47\penalty0
  (3, Part A):\penalty0 407--434, 2014.

\bibitem{barequet_2001}
G.~Barequet and S.~Har-Peled.
\newblock {Polygon containment and translation min-hausdorff-distance between
  segment sets are 3SUM-hard}.
\newblock \emph{{International Journal of Computational Geometry \&
  Applications}}, 11\penalty0 (4):\penalty0 465--474, 2001.

\bibitem{barequet_1997}
G.~Barequet, M.~Dickerson, and P.~Pau.
\newblock Translating a convex polygon to contain a maximum number of points.
\newblock \emph{{Computational Geometry: Theory and Applications}}, 8\penalty0
  (4):\penalty0 167--179, 1997.

\bibitem{chazelle_1983}
B.~Chazelle.
\newblock \emph{{Advances in Computing Research}}, volume~1, chapter The
  polygon containment problem, pages 1--33.
\newblock JAI Press, 1983.

\bibitem{dickerson_1998}
M.~Dickerson and D.~Scharstein.
\newblock Optimal placement of convex polygons to maximize point containment.
\newblock \emph{{Computational Geometry: Theory and Applications}}, 11\penalty0
  (1):\penalty0 1--16, 1998.

\bibitem{gajentaan_1995}
A.~Gajentaan and M.~H. Overmars.
\newblock {On a class of $O(n^2)$ problems in computational geometry}.
\newblock \emph{{Computational Geometry: Theory and Applications}}, 5\penalty0
  (3):\penalty0 165--185, 1995.

\bibitem{robots_92}
H.~Ishiguro, M.~Yamamoto, and S.~Tsuji.
\newblock Omni-directional stereo.
\newblock \emph{{IEEE Transactions on Pattern Analysis and Machine
  Intelligence}}, 14\penalty0 (2):\penalty0 257--262, 1992.

\bibitem{needham_2002}
T.~Needham.
\newblock \emph{Visual complex analysis}, chapter 6.II.3: A conformal map of
  the sphere, pages 283--286.
\newblock Clarendon Press, Oxford, 1998.

\bibitem{tolerance_1997}
C.~K. Yap and E.-C. Chang.
\newblock \emph{{Algorithms for Robot Motion Planning and Manipulation}},
  chapter Issues in the metrology of geometric tolerancing, pages 393--400.
\newblock A.K. Peters, Wellesley, MA, 1997.

\end{thebibliography}

\end{document}